\documentclass{llncs}
\pdfoutput=1

\usepackage[utf8]{inputenc}

\usepackage[usenames,dvipsnames]{color}
\usepackage{graphicx}
\usepackage{xspace}

\usepackage{wrapfig}
\usepackage{footnote}
\usepackage{placeins}

\usepackage{subcaption}
\usepackage{amsmath}
\usepackage{amssymb}
\usepackage{gensymb}
\usepackage{fancyhdr}%

\usepackage{enumitem}
\usepackage{pgfplots}
\usepackage{hyperref}
\usepackage[ruled,vlined,linesnumbered]{algorithm2e}
\usepackage{microtype}

\newcommand{\ie}{i.\,e.\xspace}

\newcommand{\etal}{et al.\xspace}

\newcommand{\bter}{BTER\xspace}

\newcommand{\h}[1]{\ensuremath{h(#1)}}

\newcommand{\cls}[1]{\textcolor{blue}{[CLS: #1]}\xspace}
\newcommand{\hmey}[1]{\textcolor{orange}{[HM: #1]}\xspace}
\newcommand{\mvl}[1]{\textcolor[rgb]{0,0.6,0}{[MvL: #1]}\xspace}
\newcommand{\rpr}[1]{\textcolor[rgb]{1.0,0.3,0.3}{[RP: #1]}\xspace}
\newcommand{\old}[1]{}
 \renewcommand{\cls}[1]{}
 \renewcommand{\hmey}[1]{}
 \renewcommand{\mvl}[1]{}
 \renewcommand{\rpr}[1]{}

\newcommand{\acosh}{\ensuremath\mathrm{acosh}}

\newcount\shortyear\newcount\shorthour\newcount\shortminute
\shorthour=\time\divide\shorthour by 60\shortyear=\shorthour
\multiply\shortyear by 60\shortminute=\time\advance\shortminute by
-\shortyear
\shortyear=\year\advance\shortyear by -1900

\def\zeit{\number\shorthour:\ifnum\shortminute<10 0\number\shortminute
\else\number\shortminute\fi}

\def\rightmark{{\small \today, \zeit{}}}

\def\rightmark{}    
\begin{document}

\title{Fast generation of complex networks \\ with underlying hyperbolic 
geometry\thanks{This work is partially supported by SPP 1736 \emph{Algorithms for Big Data} of the German Research Foundation (DFG) 
and  by the Ministry 
of Science, Research and the Arts Baden-Württemberg (MWK)
via project \emph{Parallel Analysis of Dynamic Networks}. }
}
\institute{\{moritz.looz-corswarem, meyerhenke, roman.prutkin, christian.staudt\}@kit.edu\\Karlsruhe Institute of Technology (KIT), Germany}
\author{Moritz von Looz \and Henning Meyerhenke \and Roman Prutkin \and Christian L. Staudt}
\date{}
\maketitle

\begin{abstract}
Complex networks have become increasingly popular for modeling various real-world phenomena. 
Realistic generative network models are
important in this context as they avoid privacy concerns of real data and simplify complex 
network research regarding data sharing, reproducibility, and scalability studies.
\emph{Random hyperbolic graphs} are a well-analyzed family of geometric graphs.
Previous work provided empirical and theoretical evidence that this generative graph model
creates networks with non-vanishing clustering and other realistic features. 
However, the investigated networks in previous applied work were small, possibly due to the quadratic
running time of a previous generator.

In this work we provide the first generation algorithm for these networks with subquadratic running time. 
We prove a time complexity of $O((n^{3/2}+m) \log n)$ with high probability for the generation process.
This running time is confirmed by experimental data with our implementation.
The acceleration stems primarily from the reduction of pairwise distance computations through
a polar quadtree, which we adapt to hyperbolic space for this purpose. 
In practice we improve the running time of a previous implementation by at least two orders of magnitude
this way. Networks with billions of edges can now be generated in a few minutes.
Finally, we evaluate the largest networks of this model published so far. 
Our empirical analysis shows that important features 
are retained over different graph densities and degree distributions.\\[1.25ex]
\end{abstract}

\section{Introduction}
The algorithmic analysis of \emph{complex networks} is 
a highly active research area since complex networks are increasingly 
used to represent phenomena as varied as the WWW, social relations, protein 
interactions, and brain topology~\cite{newman2010networks}. 
Complex networks have several non-trivial topological features: They are usually \emph{scale-free}, which refers to the presence of
a few high-degree vertices (hubs) among many low-degree vertices. 
A heavy-tail degree distribution that occurs frequently in practice follows a power law~\cite[Chap.~8.4]{newman2010networks}, \ie the number of vertices with degree $k$ is proportional to $k^{-\gamma}$, for a fixed exponent $\gamma > 0$.
Moreover, complex networks often have the \emph{small-world property}, \ie the typical distance between two vertices is surprisingly small, regardless of network size and growth.

\emph{Generative network models} play a central role in many complex network studies for several reasons:
Real data often contains confidential information; it is then desirable to work on similar synthetic networks instead.
Quick testing of algorithms requires small test cases, while benchmarks and scalability studies need bigger graphs. Graph generators can provide data at different user-defined scales for this purpose.
Also, transmitting and storing a generative model and its parameters is much easier than doing the same with a gigabyte-sized network.
A central goal for generative models is to produce networks which replicate relevant structural features of real-world networks~\cite{chakrabarti2006graph}.
Finally, generative models are an important theoretical part of network science, as they can improve our understanding of network formation.
The most widely used graph-based system benchmark in high-performance
computing, Graph500~\cite{graph500}, is based on R-MAT~\cite{chakrabarti2004r}. 
This model is efficiently computable, but has important drawbacks concerning realism
and preservation of properties over different graph sizes~\cite{KoPiPlSe14}.

\emph{Random hyperbolic graphs}, first presented by Krioukov \etal~\cite{Krioukov2010}, are a
very promising graph family in this context: They yield a provably high clustering coefficient~\cite{DBLP:conf/icalp/GugelmannPP12}, small diameter~\cite{raey} and a power-law degree distribution with adjustable exponent.
They are based on \emph{hyperbolic geometry}, which has negative curvature and is the basis for one of the three 
isotropic spaces. (The other two are Euclidean (flat) and spherical geometry (positive curvature).)
In the generative model, vertices are distributed randomly on a hyperbolic disk of radius $R$ and edges are inserted for every vertex pair whose distance is below $R$.\footnote{We consider the name ``hyperbolic unit-disk graphs'' as more precise, but we use ``random hyperbolic graphs'' 
to be consistent with the literature.}
This family of graphs has been analyzed well theoretically~\cite{bode2014probability,DBLP:conf/icalp/GugelmannPP12,raey} and Krioukov \etal~\cite{Krioukov2010} show that complex networks have a natural embedding in hyperbolic geometry.

Calculating all pairwise distances in their generation process has quadratic time complexity.
This impedes the creation of massive networks and is likely the reason previously published networks based on hyperbolic geometry have been in the range of at most $10^5$ vertices.
A faster generator is necessary to use this promising model for networks of interesting scales.
Also, more detailed parameter studies and network analytic comparisons are necessary in practice. 

\paragraph*{Outline and contribution.}
We develop, analyze, and implement a fast, subquadratic generation algorithm for random hyperbolic graphs.

To lay the foundation, Section~\ref{sec:related-work} 
introduces fundamentals of hyperbolic geometry.
The main technical part starts with Section~\ref{sec:fast-generation}, in which we use the Poincaré disk model to relate hyperbolic to Euclidean geometry.
This allows the use of a new spatial data structure, namely a polar quadtree adapted to hyperbolic space, to reduce both asymptotic complexity and running time of the generation.
We further prove the time complexity of our generation process to be 
$O((n^{3/2}+m) \log n)$ with high probability (whp, \ie $\geq 1 - 1/n$) 
for a graph with $n$ vertices, $m$ edges, and sufficiently large $n$.

In Section~\ref{sec:evaluation} we add to previous studies a comprehensive network
analytic evaluation of random hyperbolic graphs. 
This evaluation shows many realistic features over a wide parameter range.
The experimental results also confirm the theoretical expected running time. A graph with $10^7$ vertices and 
$10^9$ edges can be generated with our shared-memory parallel implementation in about 8 minutes.
The generator will be made available in a future version of NetworKit~\cite{7006796}, our open-source framework for large-scale network analysis.
\textbf{Material omitted due to space constraints can be found in the appendix.}

\section{Related Work and Preliminaries}
\label{sec:related-work}
Related generative graph models are discussed in Section~\ref{subsec:comparison-existing-generators},
where we also compare them to random hyperbolic graphs, in part by using empirical data.

\subsection{Graphs in Hyperbolic Geometry}
\label{sub:hyperbolic-introduction}
Kriokouv \etal~\cite{Krioukov2010} introduce the family of random hyperbolic graphs and show how they naturally develop a power-law degree distribution
and other properties of complex networks.
In the generative model, vertices are generated as points in polar coordinates $(\phi, r)$ on a disk of radius $R$ in the hyperbolic plane with curvature $-\zeta^2$.
We denote this disk with $\mathbb{D}_R$.
The angular coordinate $\phi$ is drawn from a uniform distribution over $[0,2\pi]$.
The probability density for the radial coordinate $r$ is given by~\cite[Eq.~(17)]{Krioukov2010} and controlled by a growth parameter $\alpha$:
\begin{equation}
 f(r) = \alpha\frac{\sinh(\alpha r)}{\cosh(\alpha R)-1}
 \label{eq:base-radial-distribution}
\end{equation}
For $\alpha=1$, this yields a uniform distribution on hyperbolic space within $\mathbb{D}_R$.
For lower values of $\alpha$, vertices are more likely to be in the center, for higher values more likely at the border of $\mathbb{D}_R$.

Any two vertices $u$ and $v$ are connected by an edge if their hyperbolic distance $\mathrm{dist}_{\mathcal{H}}(u,v)$ is below $R$.
(Kriokouv \etal also present a more general model in which edges are inserted with a probability depending on hyperbolic distance.
This extended model is not discussed here.)
The neighborhood of a point (= vertex) thus consists of the points lying in a hyperbolic circle around it.
Krioukov \etal show that for $\alpha/\zeta > \frac{1}{2}$, the resulting degree distribution follows a power law with exponent $2\cdot \alpha/\zeta +1$~\cite[Eq.~(29)]{Krioukov2010}.
Gugelmann \etal~\cite{DBLP:conf/icalp/GugelmannPP12} analyze this model theoretically and prove non-vanishing clustering and a low variation of the clustering coefficient.
Bode \etal~\cite{bode2014probability} discuss the size of the giant component and the probability that the graph is connected~\cite{raey}.
They also show~\cite{bode2014probability} that the curvature parameter $\zeta$ can be fixed while retaining all degrees of freedom, we thus assume $\zeta=1$ from now on.
The average degree $\overline{k}$ of a random hyperbolic graph is controlled with the radius $R$, using \cite[Eq.~(22)]{Krioukov2010}.

An example graph with 500 vertices, $R\approx 5.08$ and $\alpha=0.8$ is shown in Figure~\ref{fig:hyperbolic-graph-native}.
For the purpose of illustration in the figure, we add an edge $(u,v)$ where $\mathrm{dist}_{\mathcal{H}}(u,v) \leq 0.2\cdot R$.
The neighborhood of $u$ (the bold blue vertex) then consists of vertices $v$ within a hyperbolic circle (marked in blue).
\begin{figure}[bt]
\centering
\begin{subfigure}{0.45\linewidth}
\includegraphics[width=\linewidth]{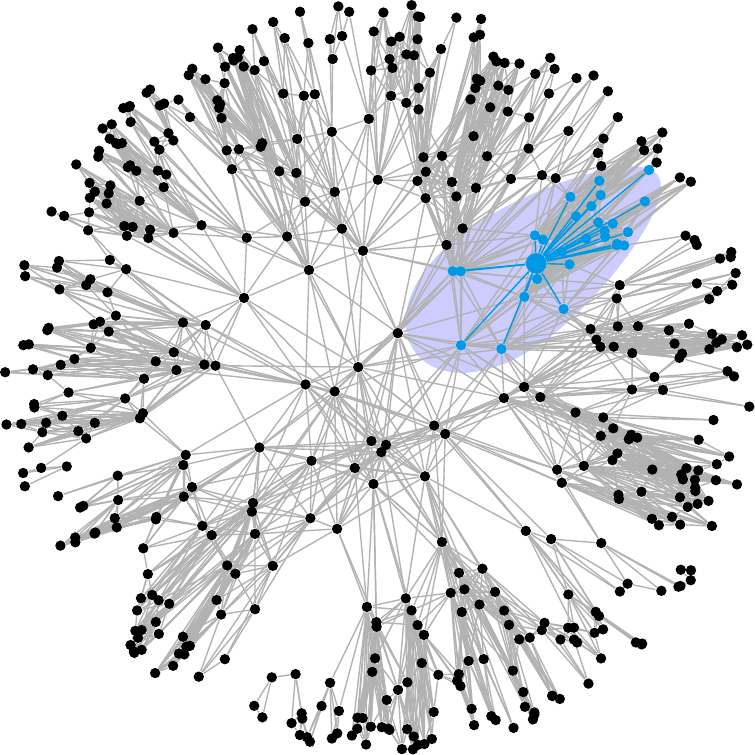}
\caption{Native representation (\cite{Krioukov2010}).}
\label{fig:hyperbolic-graph-native} 
\end{subfigure}
\quad
\begin{subfigure}{0.45\linewidth}
 \includegraphics[width=\linewidth]{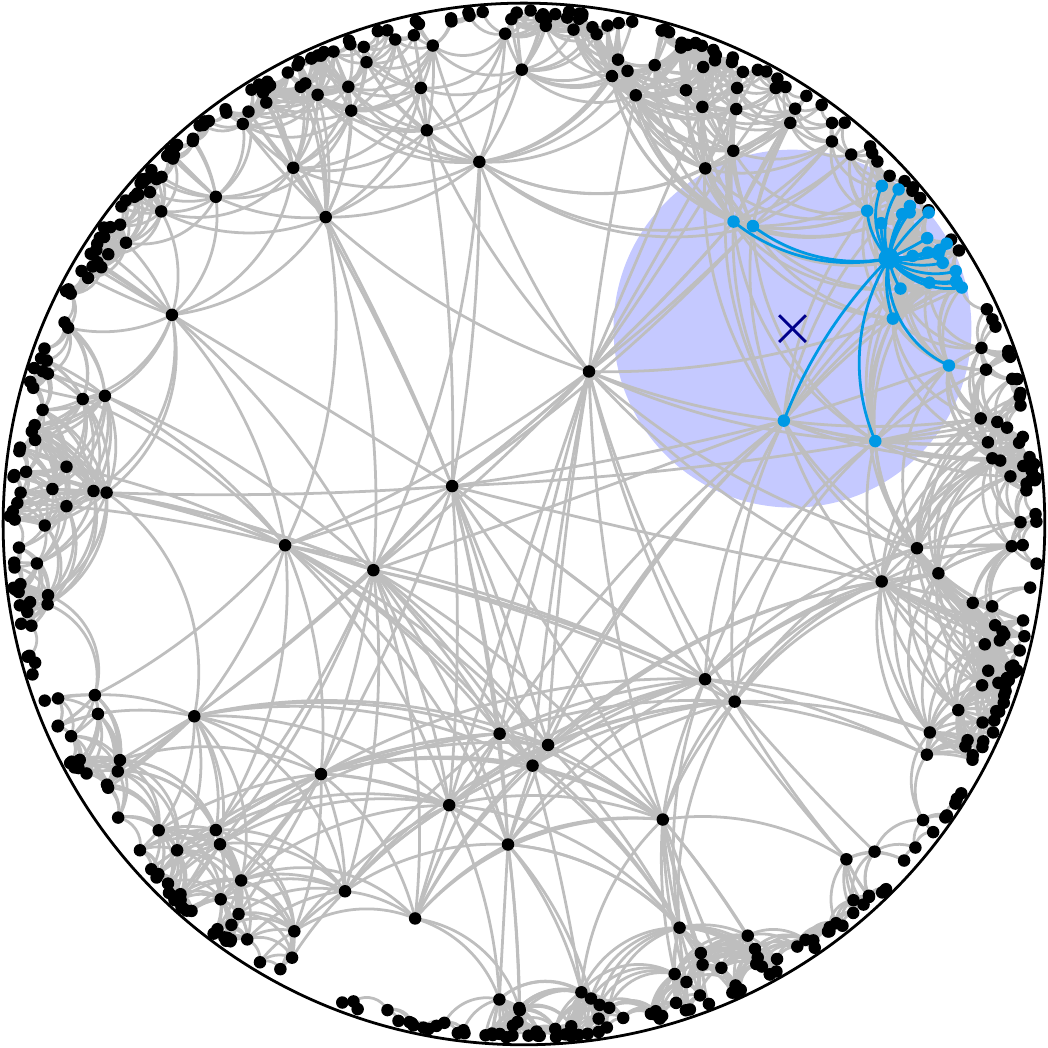}
 \caption{Poincaré disk model.}
 \label{fig:hyperbolic-graph-poincare}
\end{subfigure}
\caption{Comparison of geometries. Neighbors of the bold blue vertex are in a hyperbolic respective Euclidean circle.}
\label{fig:hyperbolic-graph-comparison}
\end{figure}

A previous generator code with a more general edge probability and quadratic time complexity is available~\cite{aldecoa2015hyperbolic}.
We show in Section~\ref{subsec:performance} that for the random hyperbolic graphs described above, our implementation is at least two orders of magnitude faster in practice.

\subsection{Poincaré Disk Model}
\label{sub:poincare}
The Poincaré disk model is one of several representations of hyperbolic space within Euclidean geometry and maps the hyperbolic plane onto the Euclidean unit disk $D_1(0)$.
The hyperbolic distance between two points $p_{E},q_{E}\in D_1(0)$ is then given by the Poincaré metric~\cite{Anderson2005}:
\begin{equation}
\mathrm{dist}_{\mathcal{H}}(p_{E}, q_{E}) = \acosh\left(1+2\frac{||p_{E}-q_{E}||^2}{(1-||p_{E}||^2)(1-||q_{E}||^2)}\right).
\label{eq:poincare-metric}
\end{equation}
Figure~\ref{fig:hyperbolic-graph-poincare} shows the same graph as in Figure~\ref{fig:hyperbolic-graph-native}, 
but translated into the Poincaré model. 
This model is conformal, \ie it preserves angles. More importantly for us, it maps hy\-per\-bolic circles onto Euclidean circles.

\section{Fast Generation of Graphs in Hyperbolic Geometry}
\label{sec:fast-generation}
We proceed by showing how to relate hyperbolic to Euclidean circles. Using this transformation, we are able to partition the Poincaré disk with a
polar quadtree that supports efficient range queries. We adapt the network generation algorithm to use this quadtree and prove subsequently that
it achieves subquadratic generation time.

\subsection{Generation Algorithm}
\label{sub:algorithm}
\paragraph{Transformation from hyperbolic geometry.}
\begin{wrapfigure}[13]{R}{0.35\linewidth}
\vspace*{-2\baselineskip}
\includegraphics[scale=0.8]{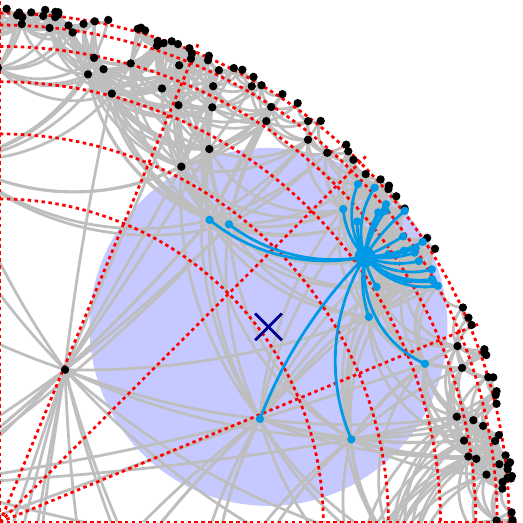}
\caption{Polar quadtree}
\label{fig:graph-polar-quadtree}
\end{wrapfigure}
Neighbors of a query point $u=(\phi_h,r_h)$ lie in a hyperbolic circle around $u$ with radius $R$.
This circle, which we denote as $H$, corresponds to a Euclidean circle $E$ in the Poincaré disk.
The center $E_c$ and radius $\mathrm{rad}_E$ of $E$ are in general different from $u$ and $R$.
All points on the boundary of $E$ in the Poincaré disk are also on the boundary of $H$ and thus have hyperbolic distance $R$ from $u$.
The points directly below and above $u$ are straightforward to construct by keeping the angular coordinate fixed and choosing 
the radial coordinates to match the hyperbolic distance:
($\phi_h, r_{e_1}$) and ($\phi_h, r_{e_2}$), with $r_{e_1}, r_{e_2} \in [0,1)$, $r_{e_1} \ne r_{e_2}$ and 
$\mathrm{dist}_{\mathcal{H}}(E_c, (\phi_h, r_{e})) = R$ for $r_e \in \{ r_{e_1}, r_{e_2}\}$.
It follows (for details see Appendix~\ref{sub:poincare-transformation}):

\begin{proposition}
$E_c$ is at ($\phi_h, \frac{2 r_h}{ab+2})$ and $\mathrm{rad}_E$ is $\sqrt{\left(\frac{2 r_h}{ab+2}\right)^2 - \frac{2 r_h^2 - ab}{ab+2}}$, with \linebreak
$a = \mathrm{cosh}(R)-1$ and $b = (1-r_h^2)$.
\label{proposition:poincare}
\end{proposition}

\paragraph{Algorithm.}
\begin{algorithm}[bt]
\KwIn{$n,\overline{k},\alpha$}
\KwOut{$G=(V,E)$}
  $R$ = getTargetRadius($n, \overline{k}, \alpha$)\tcc*{Eq.(22)\cite{Krioukov2010}}
  $V$ = $n$ vertices\;
  $T$ = empty polar quadtree of radius mapToPoincare($R$)\;
  \For{vertex $v\in V$}{
    draw $\phi[v]$ from $\mathcal{U}[0,2\pi)$\;\label{line:draw-angular}
    draw $r_{\mathcal{H}}[v]$ with density $f(r) = \alpha\sinh(r)/(\cosh(\alpha R)-1)$\label{line:draw-radial}\tcc*{Eq.(\ref{eq:base-radial-distribution})}
    $r_{E}[v]$ = mapToPoincare($r_{\mathcal{H}}[v]$)\;\label{line:poincare-mapping}
    insert $v$ into $T$ at $(\phi[v], r_{E}[v])$\;\label{line:quadtree-insertion}
  }
  \ForPar{vertex $v\in V$}{
    $C_\mathcal{H}$ = circle around $(\phi[v], r_{\mathcal{H}}[v])$ with radius $R$\; \label{line:circle-comp}
    $C_E$ = transformCircleToEuclidean($C_\mathcal{H}$)\label{line:query-circle-poincare}\tcc*{Prop. \ref{proposition:poincare}}
    \For{vertex $w\in T$.getVerticesInCircle($C_E$)}{ \label{line:neighborhood-loop}
      add $(v,w)$ to $E$\; \label{line:edge-insertion}
    }
  }
  \Return{$G$}\;
 \caption{Graph Generation}
 \label{algo:generation}
\end{algorithm}

The generation of $G=(V,E)$ with $n$ vertices and average degree $k$ is shown in Algorithm~\ref{algo:generation}.
As in previous efforts, vertex positions are generated randomly (lines \ref{line:draw-angular} and \ref{line:draw-radial}).
We then map these positions into the Poincaré disk (line~\ref{line:poincare-mapping}) and, as a new feature,
store them in a polar quadtree (line~\ref{line:quadtree-insertion}).
For each vertex $v$ the hyperbolic circle defining the neighborhood is mapped into the Poincaré disk 
according to Proposition~\ref{proposition:poincare} (lines~\ref{line:circle-comp}-\ref{line:query-circle-poincare})
-- also see Figure~\ref{fig:hyperbolic-graph-poincare}, where the neighborhood of $v$ consists of exactly the 
vertices in the light blue Euclidean circle.
Edges are then created by executing a Euclidean range query with the resulting circle in the polar quadtree
(lines~\ref{line:neighborhood-loop}-\ref{line:edge-insertion}).
The functions used by Algorithm~\ref{algo:generation} are discussed in Appendix~\ref{sec:functions}.
We use the same probability distribution for the node positions and add an edge $(u,v)$ exactly iff the hyperbolic distance between $u$ and $v$ is less than $R$.
This leads to the following proposition:

\begin{proposition}
 Algorithm~\ref{algo:generation} generates random hyperbolic graphs as defined in Section~\ref{sub:hyperbolic-introduction}.
\end{proposition}

\paragraph{Data structure.}
\label{sub:data-structure}
As mentioned above, our central data structure is a polar quadtree on the Poincaré disk.
While Euclidean quadtrees are common~\cite{Samet:2005:FMM:1076819}, we are not aware of previous adaptations to hyperbolic space.
A node in the quadtree is defined as a tuple $(\mathrm{min}_\phi, \mathrm{max}_\phi, \mathrm{min}_r, \mathrm{max}_r)$
with \(\mathrm{min}_\phi \leq \mathrm{max}_\phi\) and \(\mathrm{min}_r \leq \mathrm{max}_r\).
It is responsible for a point $p = (\phi_p, r_p) \in D_1(0)$ iff
$(\mathrm{min}_\phi \leq \phi_p < \mathrm{max}_\phi)$ and $(\mathrm{min}_r \leq r_p < \mathrm{max}_r)$.
Figure~\ref{fig:graph-polar-quadtree} shows a section of a polar quadtree, where quadtree nodes are marked by dotted red lines.
When a leaf cell is full, it is split into four children. 
Splitting in the angular direction is straightforward as the angle range is halved: $\mathrm{mid}_\phi := \frac{\mathrm{max}_\phi+\mathrm{min}_\phi}{2}$.
For the radial direction, we choose the splitting radius to result in an equal division of probability mass:

\begin{equation}
 \text{mid}_{r\mathcal{H}} := \acosh\left(\frac{\cosh(\alpha\max_{r\mathcal{H}}) + \cosh(\alpha\min_{r\mathcal{H}})}{2}\right)/\alpha
 \label{eq:splitting-radius}
\end{equation}

(Note that Eq.~(\ref{eq:splitting-radius}) uses radial coordinates in the native representation, which are converted back to coordinates in the Poincaré disk.)
This leads to two lemmas useful for establishing the time complexity of the main quadtree operations:
\begin{lemma}
Let $\mathbb{D}_R$ be a hyperbolic disk of radius $R$, $\alpha \in \mathbb{R}$, $p$ a point in $\mathbb{D}_R$ which is distributed according to Section~\ref{sub:hyperbolic-introduction}, and $T$ be a polar quadtree on $\mathbb{D}_R$.
Let $C$ be a quadtree cell at depth $i$. Then, the probability that $p$ is in $C$ is $4^{-i}$.
\label{thm:node-cell-probabilities}
\end{lemma}
\begin{lemma}
 \label{thm:quadtree-height}
Let $R$ and $\mathbb{D}_R$ be as in Lemma~\ref{thm:node-cell-probabilities}. Let $T$ be a polar quadtree on $\mathbb{D}_R$ containing $n$ points distributed according to Section~\ref{sub:hyperbolic-introduction}.
 Then, for $n$ sufficiently large, $\mathrm{height}(T) \in O(\log n)$ whp.
\end{lemma}

\subsection{Time Complexity}
\label{subsec:complexity-analysis}
The time complexity of the generator is in turn determined by the operations of the polar quadtree.

\paragraph{Quadtree Insertion.}
For the amortized analysis, we consider each element's initial and final position during the insertion of $n$ elements.
Let $\h{T}$ be the final height of quadtree $T$, let $\h{i}$ be the final level of element $i$ and let $t(i)$ be the level of $i$ when it was inserted.
During insertion of element $i$, $t(i)$ quadtree nodes are visited until the correct leaf for insertion is found, the cost for this is linear in $t(i)$.
When a leaf cell is full, it splits into four children and each element in the leaf is moved down one level.
Over the course of inserting all $n$ elements, element $i$ thus moves $\h{i} - t(i)$ times due to leaf splits.
To reach its final position at level $\h{i}$, element $i$ accrues cost of $O(t(i) + \h{i} - t(i)) = O(\h{i}) \subseteq O(\h{T})$, which is $O(\log n)$ whp due to Lemma~\ref{thm:quadtree-height}.
The amortized time complexity for a node insertion is then: $T(\mathrm{Insertion}) \in O(\log n) \text{ whp.}$

\paragraph{Quadtree Range Query.}
\newcommand{\without}{\backslash}
Neighbors of a vertex $u$ are the vertices within a Euclidean circle constructed according to Proposition~\ref{proposition:poincare}.
Let $\mathcal{N}(u)$ be this neighborhood set in the final graph, thus $\mathrm{deg}(u) := |\mathcal{N}(u)|$.
We denote leaf cells that do not have non-leaf siblings as \emph{bottom leaf cells}.
A visualization can be found in Figure~\ref{fig:visualization-bottom-leaf-cells}.

\begin{lemma}
Let $T$ and $n$ be as in Lemma~\ref{thm:quadtree-height}.
A range query on $T$ returning a point set $A$ will examine at most $O(\sqrt{n} + |A|)$ bottom leaf cells whp.
\label{lemma:bound-neighbourless-leaf-cells}
\end{lemma}

Due to Lemma~\ref{lemma:bound-neighbourless-leaf-cells}, the number of examined bottom leaf cells for a range query around $u$ is in $O(\sqrt{n} + \mathrm{deg}(u))$ whp.
The query algorithm traverses $T$ from the root downward. For each bottom leaf cell $b$, $O(\h{T})$ inner nodes and non-bottom leaf cells are examined on the path from the root to $b$.
Due to Lemma~\ref{thm:quadtree-height}, $\h{T}$ is in $O(\log n)$ whp.
The time complexity to gather the neighborhood of a vertex $u$ with degree $\text{deg}(u)$ is thus:
$T(\mathrm{RQ}(u)) \in O\left(\left(\sqrt{n}+\text{deg}(u)\right)\cdot \log n\right)  \text{ whp.}$

\paragraph{Graph Generation.}
To generate a graph $G$ from $n$ points, the $n$ positions need to be generated and inserted into the quadtree.
The time complexity of this is $n\cdot O(\log n) = O(n \log n)$ whp.
In the next step, neighbors for all points are extracted. This has a complexity of
\begin{equation}
 T(\mathrm{Edges}) = \sum_{v} O\left(\left(\sqrt{n}+\text{deg}(v)\right)\cdot \log n\right) = O\left(\left(n^{3/2}+m\right) \log n\right)  \text{ whp.}
 \label{eq:time-complexity-edges}
\end{equation}
This dominates the quadtree operations and thus total running time. We conclude:
\begin{theorem}
Generating random hyperbolic graphs can be done in $O((n^{3/2}+m) \log n)$ time whp
for sufficiently large $n$, \ie with probability $\geq 1 - 1 / n$.
\label{thm:time-complexity-graphgen}
\end{theorem}

\section{Experimental Evaluation}
\label{sec:evaluation}
We first discuss several structural properties of networks and use them to analyze random hyperbolic graphs generated with different parameters.
Comparisons to real-world networks and existing generators follow, as well as a comparison of the running time to a previous implementation~\cite{aldecoa2015hyperbolic}.

\subsection{Network Properties}
\label{subsec:properties}
We consider several graph properties characteristic of complex networks.
The \emph{degree distribution} of many complex networks follows a \emph{power law}. 
The \emph{clustering coefficient} is the fraction of closed triangles to triads (paths of length $2$) and measures how likely two vertices with a common neighbor are to be connected.
\emph{Degree assortativity} describes whether vertices have neighbors of similar degree.
A value near 1 signifies subgraphs with equal degree, a value of -1 star-like structures.
Many real networks have multiple \emph{connected components}, yet one large component is usually dominant.
\emph{k-Cores} are a generalization of components and result from iteratively peeling away vertices of degree $k$ and assigning to each vertex the \emph{core number} of the innermost core it is contained in.
The \emph{diameter} is the longest shortest path in the graph, which is often surprisingly small in complex networks.
Complex networks also often exhibit a \emph{community structure}, \ie dense subgraphs with sparse connections between them.
The fit of a given vertex partition to this structure can be quantified by \emph{modularity}~\cite{newman2010networks}.

\begin{figure}[tb]
\begin{subfigure}[t]{.3\linewidth}
\includegraphics[width=\linewidth]{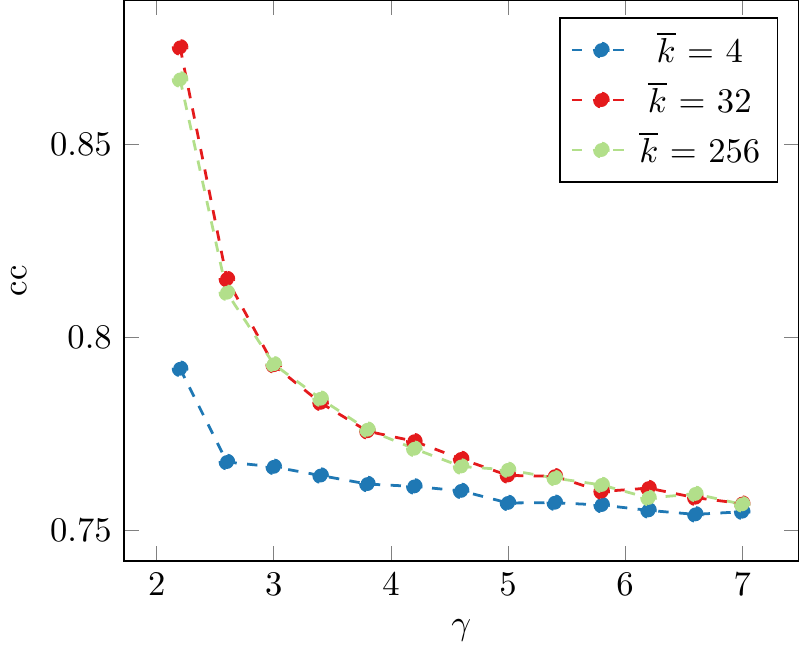}
\caption{Clustering coefficient}
\label{plot:clustercoeff}
\end{subfigure}
\quad
\begin{subfigure}[t]{.3\linewidth}
\includegraphics[width=\linewidth]{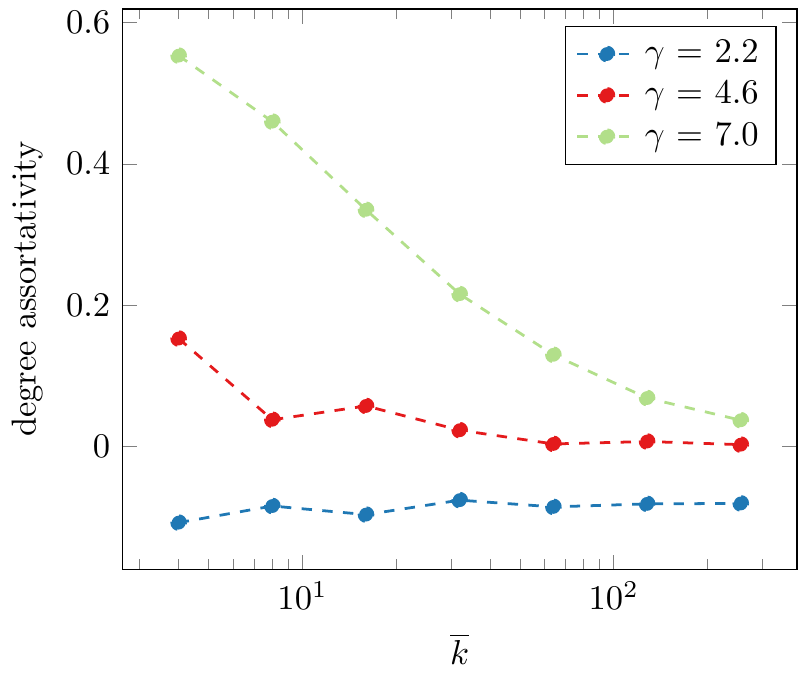}
\caption{Degree assortativity}
\label{plot:degass}
\end{subfigure}
\quad
\begin{subfigure}[t]{.3\linewidth}
\includegraphics[width=\linewidth]{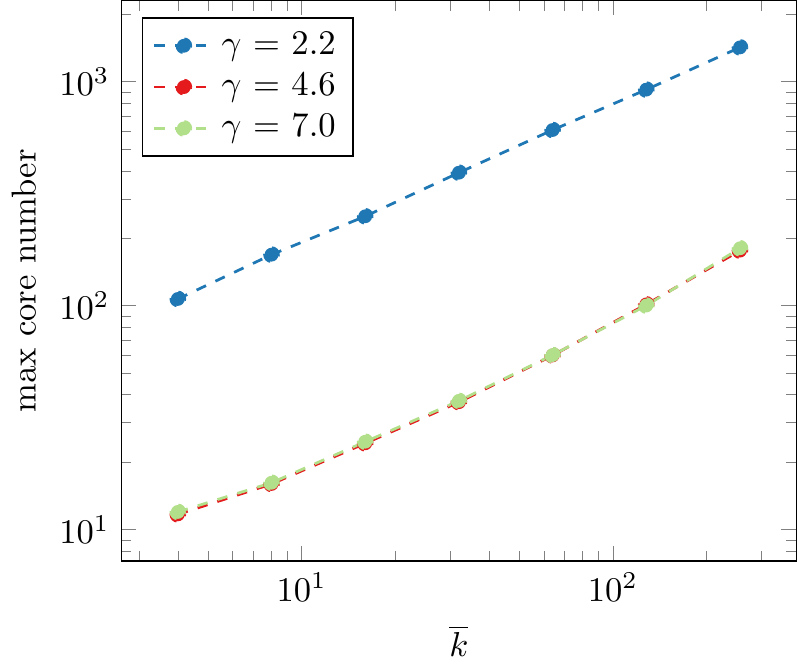}
\caption{Max core number}
\label{plot:degen}
\end{subfigure}
\caption{Properties of graphs generated with $n=10^6$, $\gamma \in [2.2,7]$ and $\overline{k} \in [4,256]$. Values are averaged over 10 runs.}
\label{plot:properties}
\end{figure}

The plots in Figures~\ref{plot:properties} and \ref{plot:additional-properties} (appendix) show the effect of average degree $\overline{k}$ and degree distribution exponent $\gamma$ on properties of the generated network.
We choose a range of $[4, 256]$ for $\overline{k}$ and $[2.2, 7]$ for $\gamma$, as these cover the ranges of real networks in Table~\ref{table:real-graphs}.
Since the geometric model inherently promotes the formation of triangles, the clustering coefficient (Fig.~\ref{plot:clustercoeff}) is between 0.6 and 0.9, significantly higher than in Erd\H{o}s-R\'{e}nyi graphs and thus more realistic for some applications.
The maximum core number corresponds closely to the density (Fig.~\ref{plot:degen}).
The graph is almost always connected for $\overline{k} \geq 32$ and almost always disconnected for $\overline{k} \leq 16$ (Fig.~\ref{plot:size-of-largest}).
The diameter of the largest component (Fig.~\ref{plot:diameter}) is highest for graphs with $\overline{k}=32$ and a high $\gamma$, as then the component is large and few high-degree hubs exist to keep the diameter small.
We use a modularity-driven algorithm~\cite{7006796} to analyze the community structure, 
the size of communities grows with the average degree (Fig.~\ref{plot:ncoms}).
Dense graphs with few communities have a relatively low modularity (Fig.~\ref{plot:modularity}).
Figure~\ref{plot:power-law-ll} shows the likelihood ratio of a power law fit to an exponential fit of the degree distribution, confirming the power-law property.
Finally, Figures~\ref{plot:properties-comparison-I} to \ref{plot:properties-comparison-III} compare graphs with $10^4$ vertices generated by our implementation and the implementation of \cite{aldecoa2015hyperbolic}.
The differences between the implementations are within the range of random fluctuations.

\subsection{Comparison with Real-world Networks}
\label{subsec:comparison-real-graphs}
We judge the realism of generated graphs by comparing them to a diverse set of
real complex networks (Table~\ref{table:real-graphs}):
\texttt{PGPgiantcompo} describes the largest connected component in the PGP web of trust, \texttt{caidaRouterLevel} and \texttt{as-Skitter} represent internet topology,
while \texttt{citationCiteseer} and \texttt{coPapersDBLP} are scientific collaboration networks,
\texttt{soc-LiveJournal} and \texttt{fb-Texas84} are social networks and \texttt{uk-2002} and \texttt{wiki\_link\_en} result from web hyperlinks.
Random hyperbolic graphs with matching $\overline{k}$ and $\gamma$ share some, but not all properties with the real networks. 
The clustering coefficient tends to be too high, with the exception of \texttt{coPapersDBLP}.
The diameter is right when matching the facebook graph, but higher by a factor of $\approx 100$ otherwise, since the geometric approach produces fewer long-range edges.
Adding 0.5\% of edges randomly to a network with average degree 10 reduces the diameter by about half while keeping the other properties comparable or unchanged,
as seen in Figures~\ref{plot:long-range-diameter} and \ref{plot:long-range-effects} in the appendix.
Just as in the real networks, the degree assortativity of generated graphs varies from slightly negative to strongly positive.
Generated dense subgraphs tend to be a tenth as large as communities typically found in real networks of the same density, and are not independent of total graph size.
Similar to real networks, random hyperbolic graphs mostly admit high-modularity partitions.
Since the most unrealistic property -- the diameter -- can be corrected with the addition of random edges, we consider random hyperbolic graphs to be reasonably realistic.

\subsection{Comparison with Existing Generators}
\label{subsec:comparison-existing-generators}

In our comparison with some existing generators, we consider realism, flexibility and performance.
Typical properties of networks generated by different generators can be found in Table~\ref{table:other-generators}.
The \emph{Barabasi-Albert model}~\cite{albert2002statistical} implements a preferential attachment process to model the growth of real complex networks. 
The probability that a new vertex will be attached to an existing vertex $v$ is proportional to $v$'s degree, which results in a power-law degree distribution.
The distribution's exponent is fixed at 3, which is roughly in the range of real-world networks. 
However, the degree assortativity is negative and the clustering coefficient low. The running time is in $\Theta(n^2)$, rendering the creation of massive networks infeasible.
The Dorogovtsev-Mendes model is designed to model network growth with a fixed average degree. It is very fast in theory ($\Theta(n)$) and practice, but at the expense of flexibility.
Clustering coefficient, degree assortativity and power law exponent of generated graphs are roughly similar to those of real-world networks.
The \emph{Recursive Matrix (R-MAT)} model~\cite{chakrabarti2004r} was proposed to recreate properties of complex networks including a power-law degree distribution, the small-world property and self-similarity.
The R-MAT generator recursively subdivides the initially empty adjacency matrix into quadrants and drops edges into it according to given probabilities.
It has $\Theta(m \log n)$ asymptotic complexity and is fast in practice.
At least the Graph500 benchmark parameters\cite{graph500}
lead to an insignificant community structure
and clustering coefficients, as no incentive to close triangles exists.
Given a degree sequence $\mathit{seq}$, the \emph{Chung-Lu (CL) model}~\cite{aiello2000random} adds edges $(u,v)$ with a probability of $p(u,v) = \frac{\mathit{seq}(u) \mathit{seq}(v)}{\sum_k \mathit{seq}(k)}$,
recreating $\mathit{seq}$ in expectation. The model can be conceived as a weighted version of the well-known Erd\H{o}s-R\'{e}nyi (ER) model and has 
similar capabilities as the R-MAT model~\cite{doi:10.1137/1.9781611972825.92}.
Implementations exist with $\Theta(n+m)$ time complexity~\cite{miller2011efficient}.
It succeeds in matching the degree distributions of the first four graphs in Table~\ref{table:real-graphs}, but in all results both clustering and degree assortativity are near zero and the diameter too small.
\begin{savenotes}
\begin{table}[tb]
\caption{Measured properties of some generative models. Parameter ranges are $n=10^6, k \in [2^2, 2^8], \gamma \in[2.2,7]$ for random hyperbolic graphs, $0.5\%$ additional long-range edges for random hyperbolic graphs with long-range edges, $\mathrm{n}=10^5,n0\in[0,10^5),k\in[0,10^4)$ for the Barabasi-Albert generator,
PGPgiantcompo, caidaRouterLevel, citationCiteseer and coPapersDBLP for Chung-Lu and \bter\ and $\mathrm{scale}=16,\mathrm{eF}=10,a\in[0.4,1.0),b\in[0,a),c\in[0,a) $ for R-MAT.}
\label{table:other-generators}
\small
\smallskip
\centering
\begin{tabular}{l|c|c|c|c|c|c|c}
 Name &		param&	m&				cc&	deg.ass.&	power law&	$\gamma$&	diam.\\
 \hline
\hline 
 \scriptsize{RHG}&	$\overline{k}, \gamma$	& $(\overline{k}/2)\cdot n$ &0.75-0.9&-0.05-0.7& 	yes&		$\approx \gamma$& 3-16k\\ 
 \scriptsize{RHG-LR}&	$\overline{k}, \gamma$	& $(\overline{k}/2)\cdot n$ &0.75-0.9&-0.05-0.4& 	yes&		$\approx \gamma$& 3-30\\ 
\hline
\scriptsize{Barabasi-Albert} &	k, n0&	$k\cdot n$&		0-0.68&			$<0$&		if $n0 < 0.3n$&$\approx	 3$&	$< 30$\\
\scriptsize{\bter} &		dd,ccd	&$\approx \sum \mathrm{dd}/2$&			matched&$\approx$matched&possible&	$\approx$matched & varies \\
\scriptsize{Chung-Lu} &	seq&	$\approx \sum \mathrm{seq}/2$&	$<10^{-2}$&		$<10^{-2}$& possible& varies&8-12\\
\scriptsize{Dor.-Men.} &		& $2\cdot n$& 		0.7&	0.02-0.05&	yes&		5-6&15-40\\
\scriptsize{R-MAT} &	$a,b,c$,eF&	$\mathrm{eF}\cdot n$&	0-1&		0-0.6&		yes&	0-10&0-18\\
\end{tabular}
\end{table}
\end{savenotes}

\emph{\bter}~\cite{KoPiPlSe14} is a two-stage structure-driven model.
It uses the standard ER model to form relatively dense subgraphs, thus forming distinct communities.
Afterwards, the CL model is used to add edges, matching the desired degree distribution in expectation~\cite{seshadhri2012community}.
This is done in  $\Theta(n+m\log d_{\mathrm{max}})$, where
$d_{\mathrm{max}}$ is the maximum vertex degree.
We test \bter with the \texttt{PGPgiantcompo}, \texttt{caidaRouterLevel}, \texttt{citationCiteseer} and 
\texttt{coPapersDBLP} networks.
The degree distributions and clustering coefficients are matched with a deviation of $\approx 5\%$.
Generated communities have a size of 5-45 on average, which is smaller than typical real communities and those of random hyperbolic graphs.

As indicated by Table~\ref{table:other-generators}, random hyperbolic graphs can match a degree distribution exponent, have stronger clustering than the Chung-Lu and R-MAT generator and are more scalable and flexible than the Barabasi-Albert generator.
Diameter (without additional random edges) and number of connected components are less realistic than those produced by \bter, but community structure is closer to typical real communities.

\subsection{Performance Measurements}
\label{subsec:performance}
\definecolor{markedcolor}{RGB}{31,120,180}
\definecolor{plottinggreen}{RGB}{178,223,138}
\definecolor{thirdhue}{RGB}{228,26,28}
\newcommand{\aconst}{3.8}
\newcommand{\bconst}{1.14}
\newcommand{\cconst}{1.38}
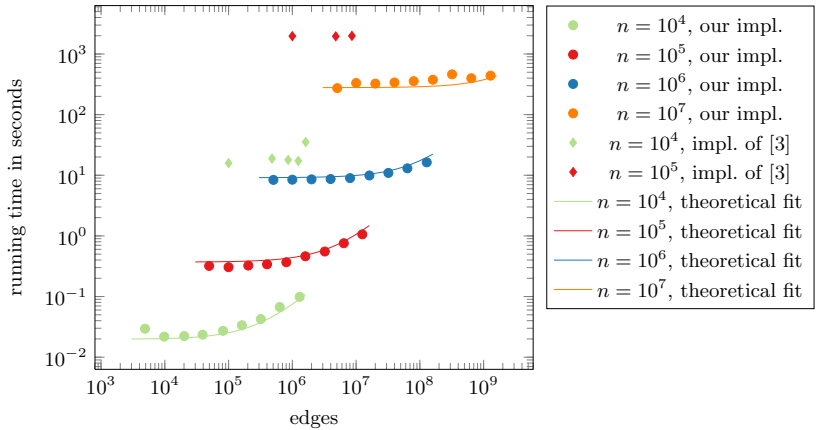
\begin{figure}[tb]
\centering
\begin{tikzpicture}[scale=0.85]
 \begin{axis}[xmode=log,ymode=log,xlabel=edges,ylabel=running time in seconds,legend entries={}, legend pos=outer north east]
  \addplot[scatter,only marks,
	   point meta = explicit symbolic,
	   scatter/classes={
	    a={draw=plottinggreen,fill=plottinggreen },%
	    b={draw=thirdhue,fill=thirdhue },%
	    c={draw=markedcolor,fill=markedcolor},
	    d={draw=orange,fill=orange},
	    g={mark=diamond*, draw=plottinggreen, fill=plottinggreen},
	    h={mark=diamond*, draw, thirdhue, fill=thirdhue}}
	    ]
	    table[meta=label]
          {plots/plotting-food-phipute-gbuild-no-leftsup-no-tbb.dat};
  \addlegendentry{$n=10^4$, our impl.};
  \addlegendentry{$n=10^5$, our impl.};
  \addlegendentry{$n=10^6$, our impl.};
  \addlegendentry{$n=10^7$, our impl.};
  \addlegendentry{$n=10^4$, impl. of \cite{aldecoa2015hyperbolic}};
  \addlegendentry{$n=10^5$, impl. of \cite{aldecoa2015hyperbolic}};
  \addplot[plottinggreen] expression[domain=3000:1600000] {\aconst *10^(-3)*4 + (\bconst *10^(-3)+\cconst *10^(-8)*x)*4};\addlegendentry{$n=10^4$, theoretical fit};
  \addplot[thirdhue] expression[domain=30000:16000000] {\aconst *10^(-2)*5 + (\bconst *10^(-4)*sqrt(10^5)+\cconst *10^(-8)*x)*5};\addlegendentry{$n=10^5$, theoretical fit};
  \addplot[markedcolor] expression[domain=300000:160000000] {\aconst *10^(-1)*6 + (\bconst *10^(-3)*sqrt(10^6)+\cconst *10^(-8)*x)*6};\addlegendentry{$n=10^6$, theoretical fit};
  \addplot[orange] expression[domain=3000000:1600000000] {\aconst *10^(0)*7 + (\bconst *10^(-2)*sqrt(10^7)+\cconst *10^(-8)*x)*7};\addlegendentry{$n=10^7$, theoretical fit};
 \end{axis}
\end{tikzpicture}
 \caption{Comparison of running times to generate networks with $10^4$-$10^7$ vertices. Circles represent running times of our implementation, diamonds the running times of the implementation of \cite{aldecoa2015hyperbolic}.
 Our running times are fitted with the equation $T(n,m) = \left(\left(\aconst\cdot10^{-7} n + \bconst\cdot10^{-9} n^{3/2} + \cconst\cdot10^{-8}m\right) \log n\right)$ seconds.
 }
 \label{plot:time-scatter}
\end{figure}
Figure~\ref{plot:time-scatter} shows the parallel running times for networks with $10^4$-$10^7$ vertices and up to $1.2\cdot 10^9$ edges.
Measurements were made on a server with 256 GB RAM and 2x8 Intel Xeon E5-2680 cores at 2.7 GHz.
We achieve a throughput of up to 13 million edges/s. Even at only $10^4$ vertices, our implementation is
two orders of magnitude faster than the implementation of~\cite{aldecoa2015hyperbolic} for the same graphs. (Note that
their implementation supports a more general model.)
Due to the smaller asymptotic complexity of $O((n^{3/2} + m)\log n)$, this gap grows with increasing graph sizes.
This complexity we prove in Section~\ref{sec:fast-generation}
is supported by the measurements, as illustrated by the lines for the theoretical fit.

\section{Conclusions}
\label{sec:conclusion}
In this work we have provided the first generator of random hyperbolic graphs with subquadratic running time.
Our parallel generator scales to large graphs that have many properties also found in real-world complex networks. The main algorithmic improvement stems from a polar quadtree, which we have adapted to hyperbolic space and which can thus be of independent interest.

The incremental quadtree construction admits a dynamic model with vertex movement, which deserves
a more thorough treatment than would have been possible here given the space constraints. It is thus part
of future work.

\begin{small}
\paragraph*{Acknowledgements.}
We thank F. Meyer auf der Heide for helpful discussions.
\end{small}

\bibliographystyle{plain}
\bibliography{Bibliography}

\begin{thebibliography}{10}

\bibitem{aiello2000random}
William Aiello, Fan Chung, and Linyuan Lu.
\newblock A random graph model for massive graphs.
\newblock In {\em Proc. 32nd ACM Symp. on Theory of computing}, pages 171--180.
  Acm, 2000.

\bibitem{albert2002statistical}
R.~Albert and A.L. Barab{\'a}si.
\newblock Statistical mechanics of complex networks.
\newblock {\em Reviews of modern physics}, 74(1):47, 2002.

\bibitem{aldecoa2015hyperbolic}
Rodrigo Aldecoa, Chiara Orsini, and Dmitri Krioukov.
\newblock Hyperbolic graph generator.
\newblock {\em arXiv preprint arXiv:1503.05180}, 2015.

\bibitem{Anderson2005}
James~W Anderson.
\newblock {\em {Hyperbolic geometry; 2nd ed.}}
\newblock Springer undergraduate mathematics series. Springer, Berlin, 2005.

\bibitem{ArratiaGordon1989}
R.~Arratia and L.~Gordon.
\newblock Tutorial on large deviations for the binomial distribution.
\newblock {\em Bulletin of Mathematical Biology}, 51(1):125--131, 1989.

\bibitem{graph500}
D.~A. Bader, Jonathan Berry, Simon Kahan, Richard Murphy, E.~Jason Riedy, and
  Jeremiah Willcock.
\newblock Graph 500 benchmark 1 ("search"), version 1.1.
\newblock Technical report, Graph 500, 2010.

\bibitem{bode2014probability}
Michel Bode, Nikolaos Fountoulakis, and Tobias M{\"u}ller.
\newblock The probability that the hyperbolic random graph is connected.
\newblock 2014.
\newblock Preprint available at
  \url{http://www.math.uu.nl/~Muell001/Papers/BFM.pdf}.

\bibitem{raey}
Michel Bode, Nikolaos Fountoulakis, and Tobias Müller.
\newblock On the giant component of random hyperbolic graphs.
\newblock In Jaroslav Nešetřil and Marco Pellegrini, editors, {\em The
  Seventh European Conference on Combinatorics, Graph Theory and Applications},
  volume~16 of {\em CRM Series}, pages 425--429. Scuola Normale Superiore,
  2013.

\bibitem{chakrabarti2006graph}
Deepayan Chakrabarti and Christos Faloutsos.
\newblock Graph mining: Laws, generators, and algorithms.
\newblock {\em ACM Computing Surveys (CSUR)}, 38(1):2, 2006.

\bibitem{chakrabarti2004r}
Deepayan Chakrabarti, Yiping Zhan, and Christos Faloutsos.
\newblock R-mat: A recursive model for graph mining.
\newblock In {\em SDM}, volume~4, pages 442--446. SIAM, 2004.

\bibitem{DBLP:conf/icalp/GugelmannPP12}
Luca Gugelmann, Konstantinos Panagiotou, and Ueli Peter.
\newblock Random hyperbolic graphs: Degree sequence and clustering - (extended
  abstract).
\newblock In {\em Automata, Languages, and Programming - 39th International
  Colloquium, {ICALP} 2012, Proceedings, Part {II}}, 2012.

\bibitem{KoPiPlSe14}
Tamara~G. Kolda, Ali Pinar, Todd Plantenga, and C.~Seshadhri.
\newblock A scalable generative graph model with community structure.
\newblock {\em SIAM J. Scientific Computing}, 36(5):C424--C452, Sep 2014.

\bibitem{Krioukov2010}
Dmitri Krioukov, Fragkiskos Papadopoulos, Maksim Kitsak, Amin Vahdat, and
  Mari\'an Bogu\~n\'a.
\newblock Hyperbolic geometry of complex networks.
\newblock {\em Physical Review E}, 82, Sep 2010.

\bibitem{miller2011efficient}
Joel~C Miller and Aric Hagberg.
\newblock Efficient generation of networks with given expected degrees.
\newblock In {\em Algorithms and Models for the Web Graph}, pages 115--126.
  Springer, 2011.

\bibitem{mitzenmacher2005probability}
Michael Mitzenmacher and Eli Upfal.
\newblock {\em Probability and computing: Randomized algorithms and
  probabilistic analysis}.
\newblock Cambridge University Press, 2005.

\bibitem{newman2010networks}
Mark Newman.
\newblock {\em Networks: an introduction}.
\newblock Oxford University Press, 2010.

\bibitem{Samet:2005:FMM:1076819}
Hanan Samet.
\newblock {\em Foundations of Multidimensional and Metric Data Structures}.
\newblock Morgan Kaufmann Publishers Inc., San Francisco, CA, USA, 2005.

\bibitem{seshadhri2012community}
C~Seshadhri, Tamara~G Kolda, and Ali Pinar.
\newblock Community structure and scale-free collections of
  {E}rd{\H{o}}s-{R}{\'e}nyi graphs.
\newblock {\em Physical Review E}, 85(5):056109, 2012.

\bibitem{doi:10.1137/1.9781611972825.92}
C.~Seshadhri, Ali Pinar, and Tamara~G. Kolda.
\newblock The similarity between stochastic kronecker and chung-lu graph
  models.
\newblock In {\em Proc. 2012 SIAM Intl. Conf. on Data Mining (SDM)}, pages
  1071--1082, 2012.

\bibitem{7006796}
C.~Staudt and H.~Meyerhenke.
\newblock Engineering parallel algorithms for community detection in massive
  networks.
\newblock {\em Parallel and Distributed Systems, IEEE Transactions on}, 2015.
\newblock doi:
  \href{http://dx.doi.org/10.1109/TPDS.2015.2390633}{10.1109/TPDS.2015.2390633}.

\end{thebibliography}

\clearpage
\appendix

\section*{Appendix}
\section{Derivation of Proposition~\ref{proposition:poincare}}
\label{sub:poincare-transformation}
\mvl{TODO: improve notation of $r_e$}
When given a hyperbolic circle with center $(\phi_h, r_h)$ and radius $\text{rad}_h$ as in Section~\ref{sub:algorithm},
the radial coordinates $r_e$ of points on the corresponding Euclidean circle can be derived with several transformations from the definition of the hyperbolic distance:
\begin{align}
\text{rad}_h &= \acosh\left(1+\frac{2(r_e-r_h)^2}{(1-r_h^2)(1-r_e^2)}\right)\\
 \Leftrightarrow \text{cosh}(\text{rad}_h)-1 &= \frac{2(r_e-r_h)^2}{(1-r_h^2)(1-r_e^2)}\\
 \Leftrightarrow (\text{cosh}(\text{rad}_h)-1)(1-r_e^2) &= \frac{2(r_e^2-2 r_h r_e + r_h^2)}{1-r_h^2}
\end{align}

To keep the notation short, we define $a := \mathrm{cosh}(\text{rad}_h)-1$ and $b := (1-r_h^2)$.
Since $\text{rad}_h > 0$ and $r_h \in [0,1)$, both $a$ and $b$ are greater than 0.
It follows:
\newcommand{\denom}{(a + \frac{2}{b})}
\begin{align}
(\text{cosh}(\text{rad}_h)-1)(1-r_e^2) &= \frac{2(r_e^2-2 r_h r_e + r_h^2)}{1-r_h^2} \\
\Leftrightarrow a - a \cdot r_e^2 &= \frac{2(r_e^2-2 r_h r_e + r_h^2)}{b} \\
 \Leftrightarrow a &= r_e^2 \cdot a + \frac{2(r_e^2-2 r_h r_e + r_h^2)}{b}\\
 \Leftrightarrow a &= r_e^2(a + \frac{2}{b}) + r_e \frac{-4 r_h}{b} + \frac{2 r_h^2}{b}\\
 \Leftrightarrow 0 &=r_e^2(a + \frac{2}{b}) + r_e \frac{-4 r_h}{b} + \frac{2 r_h^2}{b} - a\\
 \Leftrightarrow 0 &=r_e^2 + r_e \frac{-4 r_h}{b\denom} + \frac{2 r_h^2}{b\denom} - \frac{a}{\denom}
\end{align}
Solving this quadratic equation, we obtain:
\begin{equation}
 r_{e_{1,2}} = \frac{2 r_h}{ab+2}\pm\sqrt{\left(\frac{2 r_h}{ab+2}\right)^2 - \frac{2 r_h^2 - ab}{ab+2}}
\end{equation}

Since ($\phi_h, r_{e_1}$) and ($\phi_h, r_{e_2}$) are different points on the border of $E$, the center $E_c$ needs to be on the perpendicular bisector.
Its radial coordinate $r_{E_{c}}$ is thus $(r_{e_{1}}+r_{e_{2}})/2 = \frac{2 r_h}{ab+2}$.
To determine the angular coordinate, we need the following lemma:

\begin{lemma}
 Let $H$ be a hyperbolic circle with center at $(\phi_h, r_h)$ and radius $\mathrm{rad}_h$. The center $E_c$ of the corresponding Euclidean circle $E$ is on the ray from $(\phi_h, r_h)$ to the origin.
 \label{lemma:poincare-ray}
\end{lemma}
\begin{proof}
Let $p$ be a point in $H$, meaning $\mathrm{dist}_{\mathcal{H}}(p,(\phi_h, r_h)) \leq \mathrm{rad}_h$.
Let $p'$ be the mirror image of $p$ under reflection on the ray going through $(\phi_h, r_h)$ and $p$.
The point $(\phi_h, r_h)$ is on the ray and unchanged under reflection: $(\phi_h, r_h) = (\phi_h, r_h)'$.
Since reflection on the equator is an isometry in the Poincaré disk model and preserves distance, we have $\mathrm{dist}_{\mathcal{H}}(p',(\phi_h, r_h))$
= $\mathrm{dist}_{\mathcal{H}}(p,(\phi_h, r_h)') = \mathrm{dist}_{\mathcal{H}}(p,(\phi_h, r_h)) \leq \mathrm{rad}_h$ and $p'\in H$.
The Euclidean circle $E$ is then symmetric with respect to the ray and its center $E_c$ must lie on it.
\qed
\end{proof}

The radius of the circle is then derived from the distance of the center to $(\phi_h, r_{e_1})$ and 
$(\phi_h, r_{e_2})$, which is $\sqrt{\left(\frac{2 r_h}{ab+2}\right)^2 - \frac{2 r_h^2 - ab}{ab+2}}$.
With both radial and angular coordinates of $E_c$ fixed, Proposition~\ref{proposition:poincare} follows.

\section{Methods Used in Algorithm \ref{algo:generation}}
\label{sec:functions}
\subsection{getTargetRadius}
For given values of $n, \alpha$ and $R$, the expected average degree $\overline{k}$ is given by~\cite[Eq.~(22)]{Krioukov2010} and the notation $\xi=(\alpha/\zeta)/(\alpha/\zeta-1/2)$:
\begin{equation}
 \overline{k} = \frac{2}{\pi} \xi^2 n \left(e^{-\zeta R/2} + e^{-\alpha R} \left(\alpha \frac{R}{2} \left(\frac{\pi}{4} \left(\frac{\zeta}{\alpha}^2\right) -(\pi-1) \frac{\xi}{\alpha} + (\pi-2) \right) -1 \right) \right)
\end{equation}
As mentioned in Section~\ref{sec:related-work}, the value of $\zeta$ can be fixed while retaining all degrees of freedom in the model and we thus assume $\zeta=1$.
We then use binary search with fixed $n, \alpha$ and desired $\overline{k}$ to find an $R$ that gives us a close approximation of the desired average degree $\overline{k}$.

\subsection{mapToPoincare}
In the native representation\cite{Krioukov2010}, the radial coordinate $r_\mathcal{H}$ of a point $p_\mathcal{H} = (\phi_\mathcal{H}, r_\mathcal{H})$ is set to the hyperbolic distance to the origin:
\[
 r_\mathcal{H} = \mathrm{dist}_{\mathcal{\mathcal{H}}}(p_\mathcal{H},(0,0))
\]

A mapping $g: \mathcal{H}^2 \rightarrow D_1(0)$ from the native representation to the Poincaré disc model needs to preserve the hyperbolic distance to the origin across models.
By using the definition of the Poincaré metric, we can derive its radial coordinate $r_e$ in the Poincaré disc model:
\begin{equation}
g((\phi_\mathcal{H}, r_\mathcal{H})) = \left(\phi_\mathcal{H}, \sqrt{\frac{\cosh(r_\mathcal{H})-1}{\cosh(r_\mathcal{H})+1}}\right)
\end{equation}

This mapping then gives the correct hyperbolic distance with the Poincaré metric (Eq.~(\ref{eq:poincare-metric}):
\begin{align}
\mathrm{dist}_{\mathcal{H}}(g((\phi_\mathcal{H}, r_\mathcal{H})), (0,0)) &= \acosh\left(1+2\frac{||g((\phi_\mathcal{H}, r_\mathcal{H}))-(0,0)||^2}{(1-||g((\phi_\mathcal{H}, r_\mathcal{H}))||^2)(1-||(0,0)||^2)}\right)\\
&= \acosh\left(1+2\frac{||g((\phi_\mathcal{H}, r_\mathcal{H}))||^2}{(1-||g((\phi_\mathcal{H}, r_\mathcal{H}))||^2)(1)}\right)\\
&= \acosh\left(1+2\frac{||\left(\phi_\mathcal{H}, \sqrt{\frac{\cosh(r_\mathcal{H})-1}{\cosh(r_\mathcal{H})+1}}\right)||^2}{(1-||\left(\phi_\mathcal{H}, \sqrt{\frac{\cosh(r_\mathcal{H})-1}{\cosh(r_\mathcal{H})+1}}\right)||^2)}\right)\\
&= \acosh\left(1+2\frac{\left(\sqrt{\frac{\cosh(r_\mathcal{H})-1}{\cosh(r_\mathcal{H})+1}}\right)^2}{(1-\left(\sqrt{\frac{\cosh(r_\mathcal{H})-1}{\cosh(r_\mathcal{H})+1}}\right)^2)}\right)\\
&= \acosh\left(1+2\frac{\left(\frac{\cosh(r_\mathcal{H})-1}{\cosh(r_\mathcal{H})+1}\right)}{1-\left(\frac{\cosh(r_\mathcal{H})-1}{\cosh(r_\mathcal{H})+1}\right)}\right)\\
&= \acosh\left(1+2\frac{\left(\cosh(r_\mathcal{H})-1\right)}{\cosh(r_\mathcal{H})+1-\left(\cosh(r_\mathcal{H})-1\right)}\right)\\
&= \acosh\left(1+2\frac{\left(\cosh(r_\mathcal{H})-1\right)}{2}\right)\\
&= \acosh\left((\cosh(r_\mathcal{H})\right) = r_\mathcal{H}\\
\end{align}

\subsection{transformCircleToEuclidean}
The circle is constructed according to Proposition~\ref{proposition:poincare}.

\section{Proof of Lemma \ref{thm:node-cell-probabilities}}

To prove Lemma~\ref{thm:node-cell-probabilities}, we use Eq.~(\ref{eq:base-radial-distribution}) and an auxiliary lemma.

\begin{lemma}
Let $p$ be a point in $\mathbb{D}_R$ and $C$ a quadtree cell delimited by $\textnormal{min}_r$, $\textnormal{max}_r$, $\textnormal{min}_\phi$ and $\textnormal{max}_\phi$. 
The probability of $p$ being in $C$ is given by the following equation:
 \begin{equation}
 \Pr(p \in C) = \frac{\max_{\phi} - \min_{\phi}}{2\pi} \cdot \frac{\cosh(\alpha \max_{r}) - \cosh(\alpha \min_{r})}{\cosh (\alpha R) - 1} 
 \end{equation}
 \label{lemma:quadtree-cell-probability}
\end{lemma}
\begin{proof}
Let $g((\phi, r))$ be the probability density for a point $p$ at $(\phi, r)$.
In Section~\ref{sub:hyperbolic-introduction}, we used a uniform distribution over $[0,2\pi)$ for the angles and defined $f(r)$ as the probability density for radial coordinate $r$.
Since the two parameters are independent, we write:
\begin{align}
 g : [0,2\pi) \times [0, R) \rightarrow [0,1]\\
 g((\phi, r)) = \frac{1}{2\pi} \cdot f(r) = \frac{1}{2\pi} \cdot \alpha\frac{\sinh(\alpha r)}{\cosh (\alpha R) - 1}
\end{align}

With $C$ delimited by $\text{min}_r$, $\text{max}_r$, $\text{min}_\phi$ and $\text{max}_\phi$ and $g$ being the product of two independent functions, we write:
\[
\Pr(p \in C) = \int_{\min_{r}}^{\max_{r}} \int_{\min_{\phi}}^{\max_{\phi}} \frac{1}{2\pi} \cdot \alpha\frac{\sinh(\alpha r)}{\cosh (\alpha R) - 1} d\phi dr
\]
Constant factors can be moved out of the integral:
\[
\Pr(p \in C) = \frac{1}{2\pi} \cdot \frac{1}{\cosh (\alpha R) - 1} \cdot \int_{\min_{r}}^{\max_{r}} \int_{\min_{\phi}}^{\max_{\phi}} \alpha\sinh(\alpha r) d\phi dr
\]
The integrand is independent of $\phi$:
\[
\Pr(p \in C) = \frac{\max_{\phi} - \min_{\phi}}{2\pi} \cdot \frac{1}{\cosh (\alpha R) - 1} \cdot \int_{\min_{r}}^{\max_{r}} \alpha\sinh(\alpha r) dr
\]
Finally, we get:
\[
\Pr(p \in C) = \frac{\max_{\phi} - \min_{\phi}}{2\pi} \cdot \frac{\cosh(\alpha \max_{r}) - \cosh(\alpha \min_{r})}{\cosh (\alpha R) - 1} 
\]
\qed
\end{proof}
 
We proceed by proving Lemma~\ref{thm:node-cell-probabilities} by induction.
\begin{proof}
Start of induction ($i$ = 0):
At level 0, only the root cell exists and covers the whole disk.
Since $C = \mathbb{D}_R$, $\Pr(p \in C) = 1 = 4^{-0}$.

Inductive step ($i \rightarrow i+1$):
Let $C_i$ be a node at level $i$.
$C_i$ is delimited by the radial boundaries $\mathrm{min}_r$ and $\mathrm{max}_r$, as well as the angular boundaries $\mathrm{min}_\phi$ and $\mathrm{max}_\phi$.
It has four children at level $i+1$, separated by $\mathrm{mid}_r$ and $\mathrm{mid}_\phi$. Let $SW$ be the south west child of $C_i$.
With Lemma~\ref{lemma:quadtree-cell-probability}, the probability of $p\in SW$ is:
\begin{equation}
\Pr(p \in SW) = \frac{\mathrm{mid}_{\phi} - \min_{\phi}}{2\pi} \cdot \frac{\cosh(\alpha \mathrm{mid}_r) - \cosh(\alpha \min_{r})}{\cosh (\alpha R) - 1} 
 \end{equation}
The angular range is halved ($\mathrm{mid}_\phi := \frac{\mathrm{max}_\phi+\mathrm{min}_\phi}{2}$) and $\mathrm{mid}_r$ is selected according to Eq.~(\ref{eq:splitting-radius}):
\begin{equation*}
 \text{mid}_r := \acosh\left(\frac{\cosh(\alpha\max_r) + \cosh(\alpha\min_r)}{2}\right)/\alpha
\end{equation*}

This results in a probability of 
\begin{align}
&\frac{\frac{\max_\phi + \min_\phi}{2} - \min_{\phi}}{2\pi} \cdot \frac{\cosh(\alpha \cdot \acosh\left(\frac{\cosh(\alpha\max_r) + \cosh(\alpha\min_r)}{2}\right)/\alpha) - \cosh(\alpha \min_{r})}{\cosh (\alpha R) - 1}\\
&= \frac{\max_\phi + \min_\phi - 2\min_{\phi}}{4\pi} \cdot \frac{\cosh(\acosh\left(\frac{\cosh(\alpha\max_r) + \cosh(\alpha\min_r)}{2}\right)) - \cosh(\alpha \min_{r})}{\cosh (\alpha R) - 1}\\
&= \frac{\max_\phi - \min_{\phi}}{4\pi} \cdot \frac{\frac{\cosh(\alpha\max_r) + \cosh(\alpha\min_r)}{2} - \cosh(\alpha \min_{r})}{\cosh (\alpha R) - 1}\\
&= \frac{\max_\phi - \min_{\phi}}{4\pi} \cdot \frac{\cosh(\alpha\max_r) + \cosh(\alpha\min_r) - 2\cosh(\alpha \min_{r})}{2(\cosh (\alpha R) - 1)}\\
&= \frac{1}{4}\frac{\max_\phi - \min_{\phi}}{2\pi} \cdot \frac{\cosh(\alpha\max_r) - \cosh(\alpha\min_r)}{\cosh (\alpha R) - 1}\\
&= \frac{1}{4} \Pr(p \in C_i)
\end{align}
As per the induction hypothesis, $\Pr(p \in C_i)$ is $4^{-i}$ and $\Pr(p \in SW)$ is thus $\frac{1}{4}\cdot 4^{-i} = 4^{-(i+1)}$.
Due to symmetry when selecting $\mathrm{mid}_\phi$, the same holds for the south east child of $C_i$. Together, they contain half of the probability mass of $C_i$.
Again due to symmetry, the same proof then holds for the northern children as well.
\qed
\end{proof}

\section{Proof of Lemma~\ref{thm:quadtree-height}}
We say ``with high probability'' when referring to a probability of at least $1 - 1/n$ (for sufficiently large $n$).
While previous results exist for the height and cost of two-dimensional quadtrees~\cite{Samet:2005:FMM:1076819}, these quadtrees differ from our polar hyperbolic approach in important properties and the results are
not easily transferable. For example, we adjust the size of our quadtree cells to result in an equal division of probability mass when taking the hyperbolic geometry into account, see Lemma~\ref{lemma:quadtree-cell-probability}.
We thus make use of a lemma from the theory of \emph{balls into bins} instead:

\begin{lemma}[\cite{mitzenmacher2005probability}]
\label{lem:balls-into-bins}
When $n$ balls are thrown independently and uniformly at random into $n$ bins, the probability that the 
maximum load is more than $3 \ln n / \ln \ln n$ is at most $1/n$ for $n$ sufficiently large.
\end{lemma}

\begin{proof}[of Lemma~\ref{thm:quadtree-height}]
In a complete quadtree, $4^k$ cells exist at height $k$. For analysis purposes only, we construct such 
a complete but initially empty quadtree of height $k = \lceil \log_4(n)\rceil$, which has at least $n$ leaf cells.
As seen in Lemma~\ref{thm:node-cell-probabilities}, a given point has an equal chance to land in each leaf cell.
Hence, we can apply Lemma~\ref{lem:balls-into-bins} with each leaf cell being a bin and a point being a ball.
(The fact that we can have more than $n$ leaf cells only helps in reducing the average load.)
From this we can conclude that, for $n$ sufficiently large, 
no leaf cell of the current tree contains more than $3 \ln n / \ln \ln n$ points with high probability (whp). 
Even if we had to construct a subtree below a current leaf $l$ to store points whose number exceeds the
capacity of $l$, the height of this subtree cannot exceed the number of points in the corresponding area,
which is at most $3 \ln n / \ln \ln n$ whp. Consequently, the total quadtree height does not 
exceed $O(\log n)$ whp.

Let $T'$ be the quadtree as constructed in the previous paragraph, starting with a complete quadtree of height~$k$ and splitting leaves when their capacity is exceeded.
Let $T$ be the quadtree created in our algorithm, starting with a root node, inserting points and also splitting leaves when necessary, growing the tree downward.

Since both trees grow downward as necessary to accomodate all points, but $T$ does not start with a complete quadtree of height~$k$, the set of quadtree nodes in $T$ is a subset of the quadtree nodes in $T'$.
Consequently, the height of $T$ is bounded by $O(\log n)$ whp as well.
\qed
\end{proof}

\section{Proof of Lemma~\ref{lemma:bound-neighbourless-leaf-cells}}
\begin{proof}
As done previously, we denote leaf cells that do not have non-leaf siblings as \emph{bottom leaf cells}, see Figure~\ref{fig:visualization-bottom-leaf-cells} for an example.
The following proof is done for a leaf capacity of one.
Since a larger leaf capacity does not increase the tree height and adds only a constant factor to the cost of examining a leaf, this choice does not result in a loss of generality.

Let $L$ be the set of bottom leaf cells containing a vertex in $A$ and let $Q$ be the set of bottom leaf cells examined by the range query.
Since the contents of leaf cells are disjoint, $|L| \leq |A|$ holds.
The set $Q \without L$ consists of leaf cells which are examined by the range query but yield no points in $A$.
These are empty leaf cells within the query circle as well as cells cut by the circle boundary.

Empty leaf cells occur when a previous leaf cell is split since its capacity is exceeded by at least one point.
Therefore an empty leaf cell $a$ in the interior of the query circle only occurs when at least two points happened to be allocated to its parent cell $b$.
A split caused by two points creates four leaf cells, therefore there are at most twice as many empty bottom leaf cells as points.

The number of cells cut by the boundary can be bounded with a geometric argument.
On level $k = \lceil\log_4 n \rceil$, at most $4^k$ cells exist, defined by at most $2^{k}$ angular and $2^{k}$ radial divisions.
When following the circumference of a query circle, each newly cut leaf cell requires the crossing of an angular or radial division.
Each radial and angular coordinate occurs at most twice on the circle boundary, thus each division can be crossed at most twice.
With two types of divisions, the circle boundary crosses at most $2\cdot2\cdot 2^k = 4\cdot2^{\lceil\log_4 n \rceil}$ cells on level $k$.
Since the value of $4\cdot2^{\lceil\log_4 n \rceil}$ is smaller than $4\cdot2^{1+\log_4 n}$, this yields $< 8\cdot \sqrt{n}$ cut cells.

In a balanced tree, all cells on level $k$ are leaf cells and the bound calculated above is an upper bound for $|Q \without L|$.
For the general case of an unbalanced tree, we use auxiliary Lemma~\ref{lemma:leaf-cell-descendants-bound}, which bounds to $O(\sqrt{n})$ the number of bottom leaf cells descendant from cells cut in level $k$.
\qed
\end{proof} 

\begin{figure}
\centering
 \includegraphics[width=.45\linewidth]{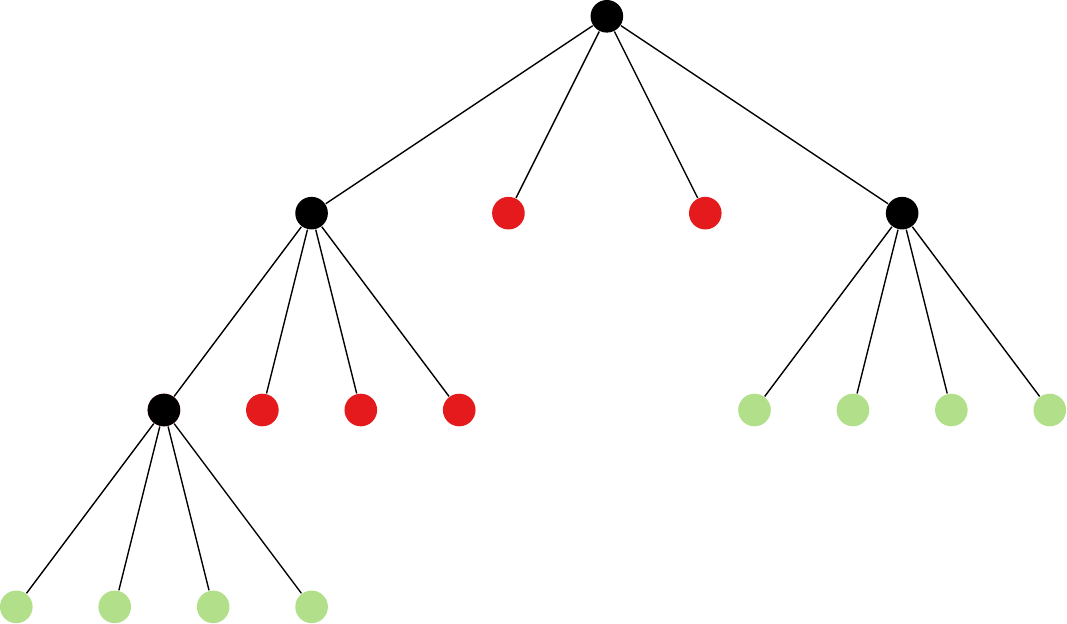}
 \caption{Visualization of bottom leaf cells in a quadtree. Bottom leaf cells are marked in green, non-bottom leaf cells in red and interior cells in black.}
 \label{fig:visualization-bottom-leaf-cells}
\end{figure}

\begin{lemma}
 Let $T$, $n$, $R$, $\mathbb{D}_R$ and $\alpha$ be as in Lemma~\ref{thm:quadtree-height}.
 Let $k := \lceil\log_4 n\rceil$ and let $C$ be a set of $\lfloor c\cdot\sqrt{n}\rfloor$ quadtree cells at level $k$, for $c \geq 1$.
 The total number of bottom leaf cells among the descendants of cells in $C$ is then bounded by $4c\cdot\sqrt{n}$ whp.
 \label{lemma:leaf-cell-descendants-bound}
\end{lemma}
\begin{proof}
New leaf cells are only created if a point is inserted in an already full cell.
We argue similarly to Lemma~\ref{lemma:bound-neighbourless-leaf-cells} that the descendants contain at most twice as many empty bottom leaf cells as points.
The number of points in the cells of $C$ is a random variable, which we denote by $X$.
Since each point position is drawn independently and is equally likely to land in each cell at a given level (Lemma~\ref{lemma:quadtree-cell-probability}), $X$ follows a binomial distribution:
\begin{equation}
X \sim B\left(n, \frac{\lfloor c\cdot \sqrt{n}\rfloor}{4^k}\right)
\end{equation}
For ease of calculation, we define a slightly different binomial distribution $Y \sim B(n, \frac{c\cdot \sqrt{n}}{n})$. Since $n \leq 4^k$ and $c\cdot \sqrt{n} \geq \lfloor c\cdot \sqrt{n}\rfloor$, the tail bounds calculated for $Y$ also hold for $X$.

Let $H(\frac{2c\cdot \sqrt{n}}{n},\frac{c\cdot \sqrt{n}}{n})$ be the relative entropy (also known as Kullback-Leibler divergence) of the two Bernoulli distributions $B(\frac{2c\cdot \sqrt{n}}{n})$ and $B(\frac{c\cdot \sqrt{n}}{n})$.
We can then use a tail bound from \cite{ArratiaGordon1989} to gain an upper bound for the probability that more than $2c$ points are in the cells of $C$:
\begin{equation}
 \Pr(Y \geq 2c\cdot\sqrt{n}) \leq \exp\left(-nH\left(\frac{2c\cdot\sqrt{n}}{n},\frac{c\cdot\sqrt{n}}{n}\right)\right)
 \label{eq:binomial-tail-bound}
\end{equation}

For consistency with our previous definition of ``with high probability'', we need to show that $\Pr(Y \geq 2c\sqrt{n}) \leq 1/n$ for $n$ sufficiently large.
To do this, we interpret $\Pr(Y \geq 2c\cdot\sqrt{n}) / (1/n)$ as an infinite sequence and observe its behavior for $n \rightarrow \infty$.
Let $a_n := \Pr(Y \geq 2c\cdot\sqrt{n}) / (1/n) = n\cdot \Pr(Y \geq 2c\cdot\sqrt{n})$ and $b_n := n\cdot\exp\left(-nH\left(\frac{2c\cdot\sqrt{n}}{n},\frac{c\cdot\sqrt{n}}{n}\right)\right)$.
From Eq.~\eqref{eq:binomial-tail-bound} we know that $a_n \leq b_n$.

Using the definition of relative entropy, we iterate over the two cases (point within $C$, point not in $C$) for both Bernoulli distributions and get:
\begin{align}
b_n &= n\cdot \exp\left(-n H\left(\frac{2c \sqrt{n}}{n},\frac{c \sqrt{n}}{n}\right) \right) \\
 &= n\cdot \exp\left(-n \left(\left(\frac{2c}{\sqrt{n}}\right)\ln 2 + \left(1-\frac{2c}{\sqrt{n}}\right) \ln \frac{1 - \frac{2c \sqrt{n}}{n}}{1-\frac{c \sqrt{n}}{n}}\right)\right)\\
 &= n\cdot \exp\left(-n \frac{2c}{\sqrt{n}}\ln 2\right) \cdot \exp\left(-n \left(1-\frac{2c}{\sqrt{n}}\right) \ln \frac{\sqrt{n}-2c}{\sqrt{n}-c}\right)\\
 &= n\cdot \exp\left(-2c\sqrt{n}\ln 2\right) \cdot \exp\left(\left(n-2c\sqrt{n}\right) \ln \frac{\sqrt{n}-c}{\sqrt{n}-2c}\right)\\
 &= n\cdot \frac{1}{2^{2c\sqrt{n}}} \cdot \left(\frac{\sqrt{n}-c}{\sqrt{n}-2c}\right)^{n-2c\sqrt{n}}\\
 &= n\cdot \frac{1}{4^{c\sqrt{n}}} \cdot \left(1+ \frac{c}{\sqrt{n}-2c}\right)^{n-2c\sqrt{n}}
\end{align}

(While $b_n$ is undefined for $n \in \{c^2, 4c^2\}$, we only consider \emph{sufficiently large $n$} from the outset.)

We apply a variant of the root test and consider the limit $\lim_{n \rightarrow \infty} (b_n)^{\frac{1}{\sqrt{n}}}$ for an auxiliary result:
\begin{align}
 &\lim_{n \rightarrow \infty} \left(n\cdot \frac{1}{4^{c\sqrt{n}}} \cdot \left(1+ \frac{c}{\sqrt{n}-2c}\right)^{n-2c\sqrt{n}}\right)^{\frac{1}{\sqrt{n}}}\\
 =& \lim_{n \rightarrow \infty} n^{\frac{1}{\sqrt{n}}}\cdot \frac{1}{4^{c}} \cdot \left(1+ \frac{c}{\sqrt{n}-2c}\right)^{\sqrt{n}}\left(1+ \frac{c}{\sqrt{n}-2c}\right)^{-2c}\\
 =& 1\cdot \frac{1}{4^{c}} \cdot e^c\cdot 1 = \left(\frac{e}{4}\right)^c
\end{align}

From $e/4 < 0.7$, $c \geq 1$ and the limit definition, it follows that almost all elements in $(b_n)^{\frac{1}{\sqrt{n}}}$ are smaller than 0.7 and thus almost all elements in $b_n$ are smaller than $0.7^{\sqrt{n}}$.
Thus $\lim_{n \rightarrow \infty} b_n \leq \lim_{n \rightarrow \infty} 0.7^{\sqrt{n}} = 0$.
Due to Eq.~\eqref{eq:binomial-tail-bound}, we know that $a_n$ is smaller than $b_n$ for large $n$,
and therefore that the number of points in $C$ is smaller than $2c\cdot \sqrt{n}$ with probability at most $1/n$ for $n$ sufficiently large.
Again with high probability, this limits the number of non-leaf cells in $C$ to $c\cdot \sqrt{n}$ and thus the number of bottom leaf cells to $4c\cdot \sqrt{n}$, proving the claim.
\qed
\end{proof}

\pagebreak
\section{Further Graph Property Parameter Studies}
\begin{figure}[h]
\begin{subfigure}[t]{.3\linewidth}
\includegraphics[width=\linewidth]{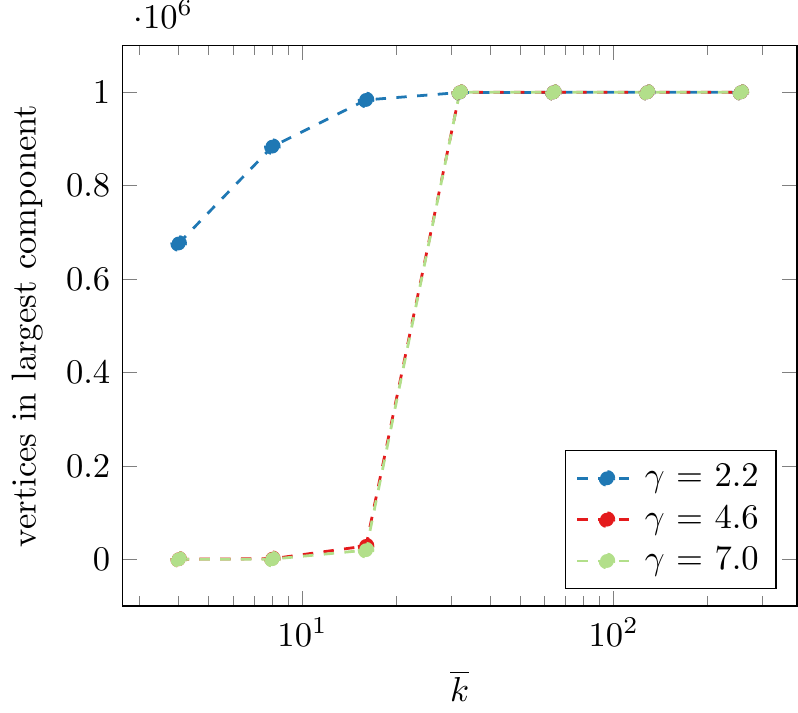}
\caption{Vertices in largest component}
\label{plot:size-of-largest}
\end{subfigure}
\quad
\begin{subfigure}[t]{.3\linewidth}
\includegraphics[width=\linewidth]{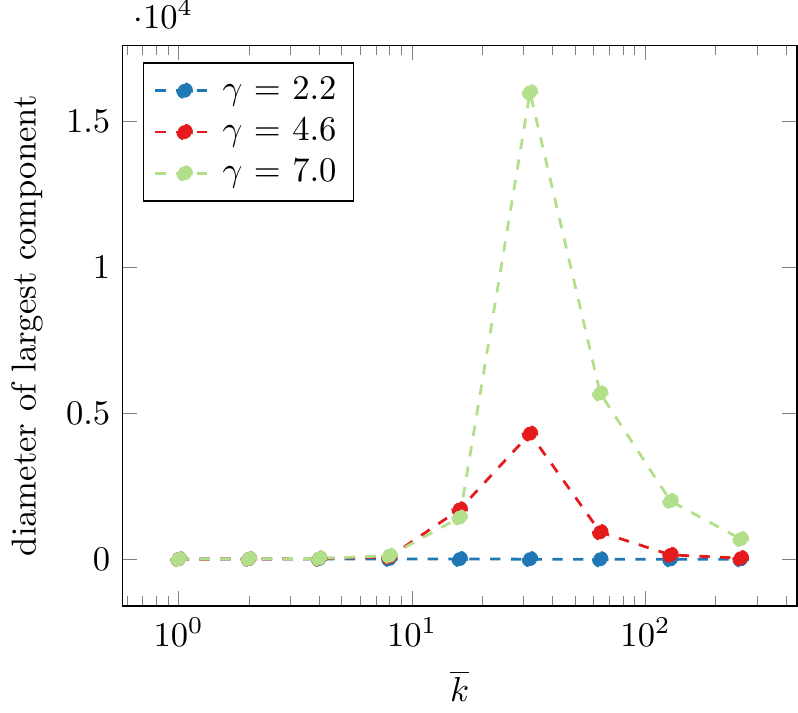}
\caption{Diameter of largest component}
\label{plot:diameter}
\end{subfigure}
\quad
\begin{subfigure}[t]{.3\linewidth}
\includegraphics[width=\linewidth]{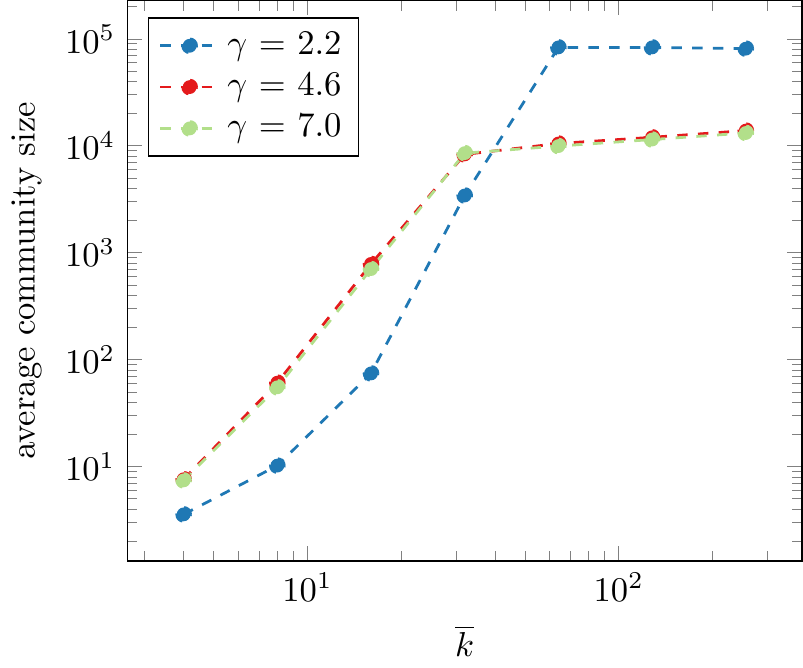}
\caption{Avg. community size}
\label{plot:ncoms}
\end{subfigure}

\begin{subfigure}[t]{.3\linewidth}
\includegraphics[width=\linewidth]{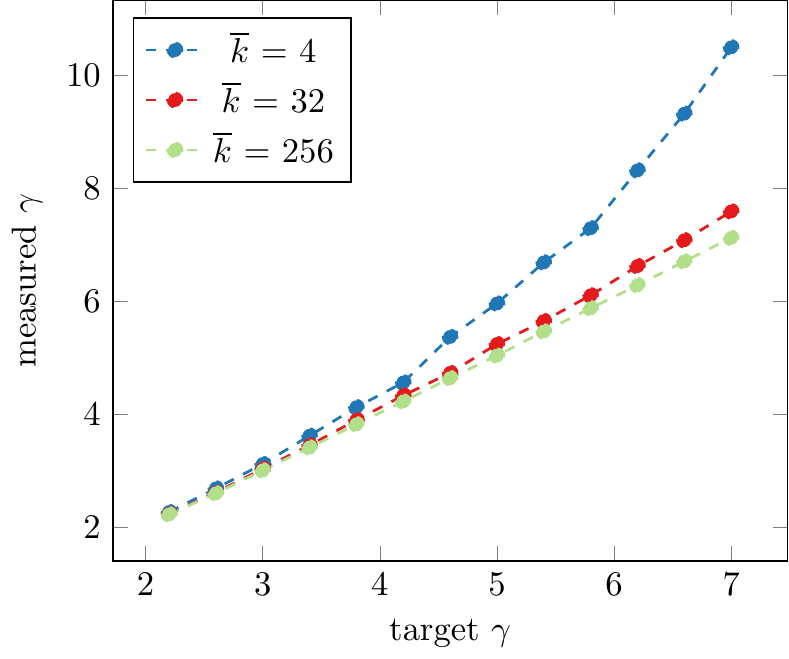}
\caption{Measured vs. desired $\gamma$}
\label{plot:gamma}
\end{subfigure}
\begin{subfigure}[t]{.3\linewidth}
\includegraphics[width=\linewidth]{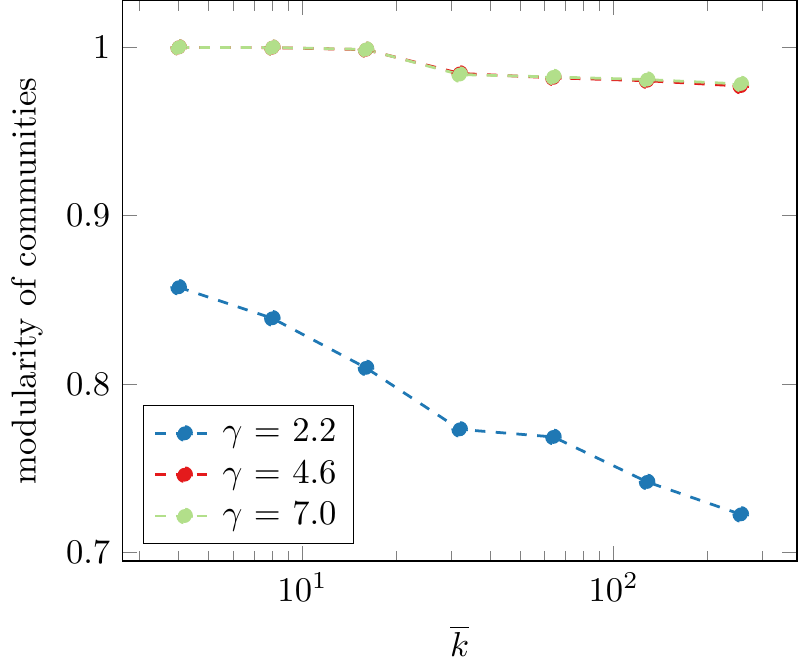}
\caption{Modularity of communities.}
\label{plot:modularity}
\end{subfigure}
\quad
\begin{subfigure}[t]{.3\linewidth}
\includegraphics[width=\linewidth]{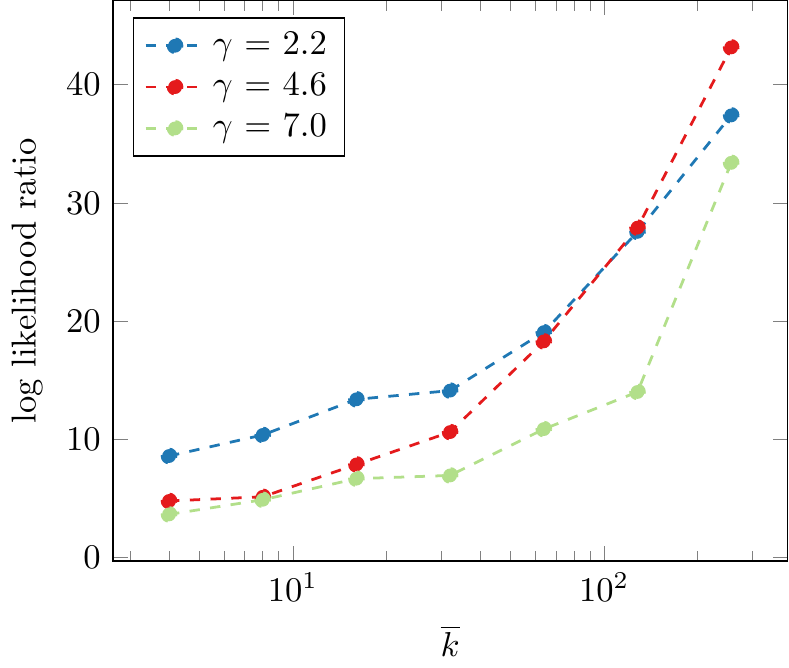}
\caption{Loglikelihood of power-law degree distribution}
\label{plot:power-law-ll}
\end{subfigure}
\caption{Further parameter studies, omitted from Figure~\ref{plot:properties} due to space constraints. Values are averaged over 10 runs, except for the diameter.}
\label{plot:additional-properties}
\end{figure}
\FloatBarrier
\pagebreak
\section{Properties of Some Real Networks}
\begin{savenotes}
\begin{table}
\caption{Properties of various real networks. The columns show the number of vertices and edges, the clustering coefficient, the maximum core number, the log likelihood of a power-law degree distribution,
the exponent of an optimal power-law fit, the degree assortativity, the diameter, average size of communities and modularity of the community structure.}
\label{table:real-graphs}
\smallskip 
\small
\centering
\begin{tabular}{l|c|c|c|c|c|c|c|c|c|c}
  &		n&	m&	cc&	max core&pl ll&	$\gamma$&	deg.ass.&	diameter&	comm. size&	mod.\\
 \hline
\scriptsize{PGPgiantcompo} &	10K&	24K&	0.44&	31&	2.04&	4.41&	0.23&	24&		101&		0.88\\
\scriptsize{fb-Texas84} &	36K&	1.6M&	0.19&	81&	1.54&	4.8&	0&	7-8&		1894&		0.38\\
\scriptsize{caidaRouterLevel} & 192K&	609K&	0.19&	32&	6.73&	3.46&	0.02&	26-30&		365&		0.85\\
\scriptsize{citationCiteseer} &268K&	1M&	0.21&	15&	9.6&	3.0&	-0.05&	36-40&		1861&		0.80\\
\scriptsize{coPapersDBLP} &	540K&	15M&	0.81&	336&	4.04&	5.95&	0.50&	23-24&		2842&		0.84\\
\scriptsize{as-Skitter} &	1.7M&	11M&	0.3&	111&	20.3&	2.35&	-0.08&	31-40&		1349&		0.83\\
\scriptsize{soc-LiveJournal} &4.8M&	43M&	0.36&	373&	6.94&	3.34&	0.02&	19-24&		632&		0.75\\
\scriptsize{uk-2002} &		18.5M&	261M&	0.69&	943&	331&	2.45&	-0.02&	45-48&		441&		0.98\\
\scriptsize{wiki\_link\_en} & 	27M & 	547M & 	0.10&	122&	26&	3.41&	-0.05&	-& 		21.6&		0.67\\
\end{tabular}
\end{table}
\end{savenotes}
\FloatBarrier
\section{Comparison with Previous Implementation\cite{aldecoa2015hyperbolic}}
Both implementations sample random graphs, making a direct comparison of generated graphs difficult.
In its output files, the implementation of \cite{aldecoa2015hyperbolic} provides the hyperbolic coordinates of the generated points.
Yet, since the distance threshold $R$ is computed non-deterministically with a Monte Carlo process and not written to the log file, we do not have all necessary information to recreate the graphs exactly.
The Figures~\ref{plot:properties-comparison-I}, \ref{plot:properties-comparison-II} and \ref{plot:properties-comparison-III} show properties of the generated graphs instead, averaged over 10 runs.
Plots showing graphs created with the implementation of \cite{aldecoa2015hyperbolic} are on the left, plots created with our implementation are on the right.
Some random fluctuations are visible, but for almost all properties the averages of our implementation are very similar to the implementation of \cite{aldecoa2015hyperbolic}.
The measured values of $\gamma$ for thin graphs and various target $\gamma$s differs from the previous implementation, but the fluctuation within the measurements of each implementations are sufficiently strong that it leads us to assume some measurement noise.
The differences between the implementations are smaller than the variations within one implementation.
\begin{figure}
\begin{subfigure}[t]{.45\linewidth}
\includegraphics[width=\linewidth]{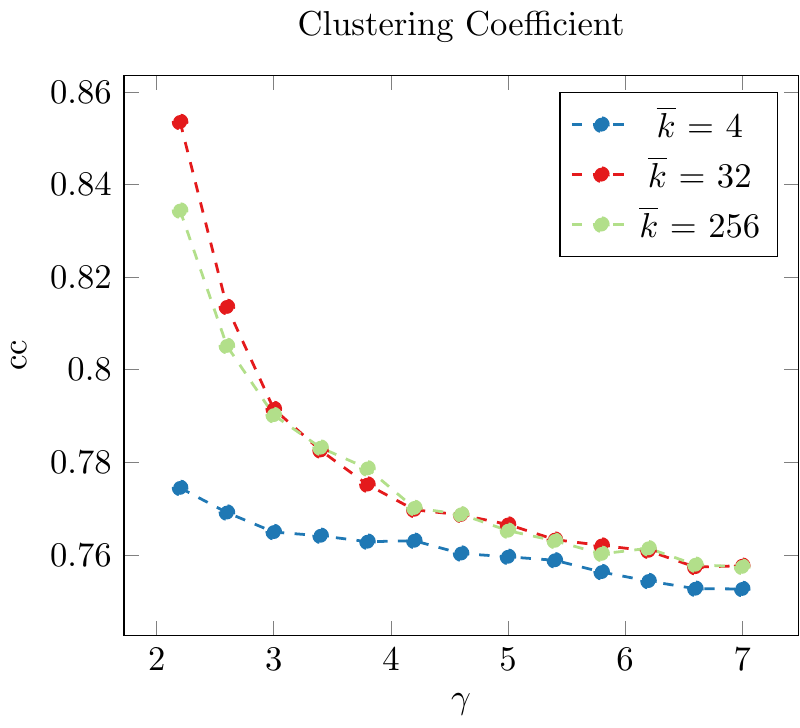}
\label{plot:native-clustercoeff}
\end{subfigure}
\quad
\begin{subfigure}[t]{.45\linewidth}
\includegraphics[width=\linewidth]{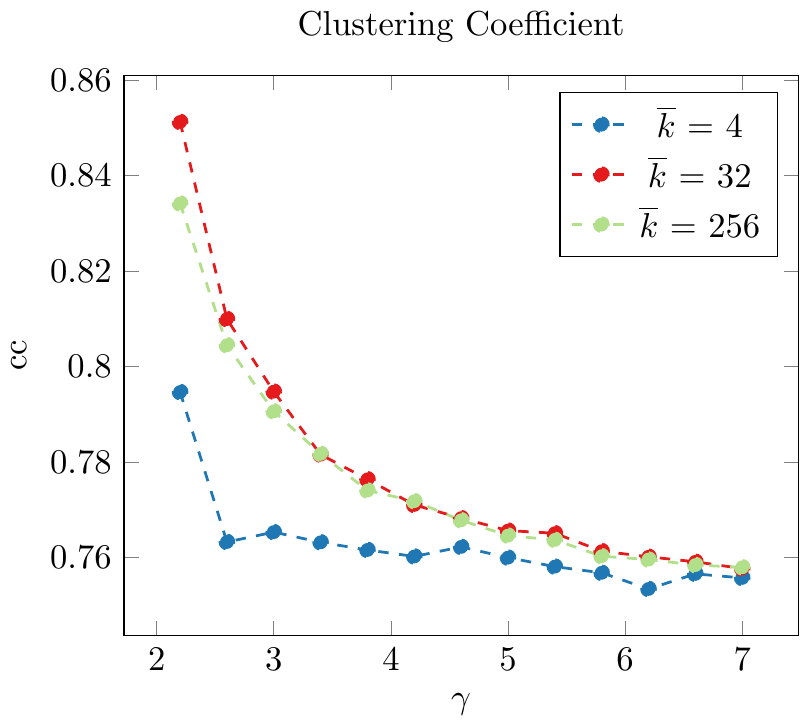}
\label{plot:comparison-clustercoeff}
\end{subfigure}

\begin{subfigure}[t]{.45\linewidth}
\includegraphics[width=\linewidth]{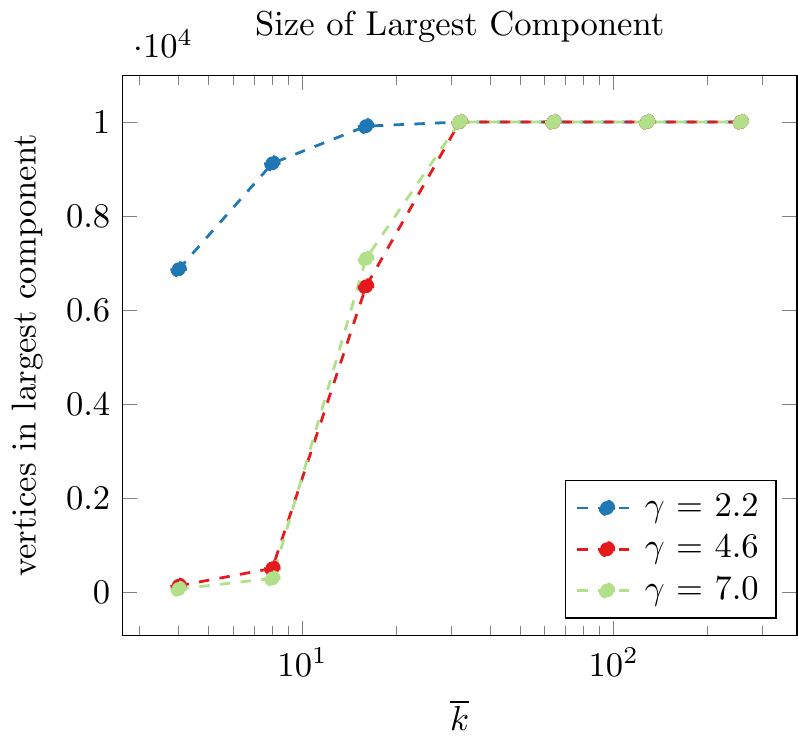}
\label{plot:native-size-of-largest}
\end{subfigure}
\quad
\begin{subfigure}[t]{.45\linewidth}
\includegraphics[width=\linewidth]{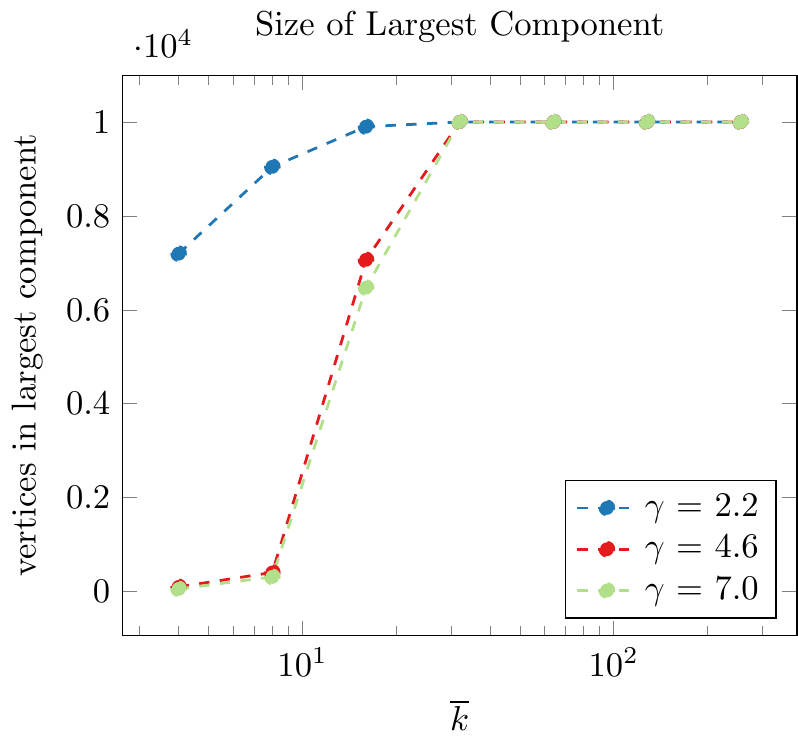}
\label{plot:comparison-size-of-largest}
\end{subfigure}

\begin{subfigure}[t]{.45\linewidth}
\includegraphics[width=\linewidth]{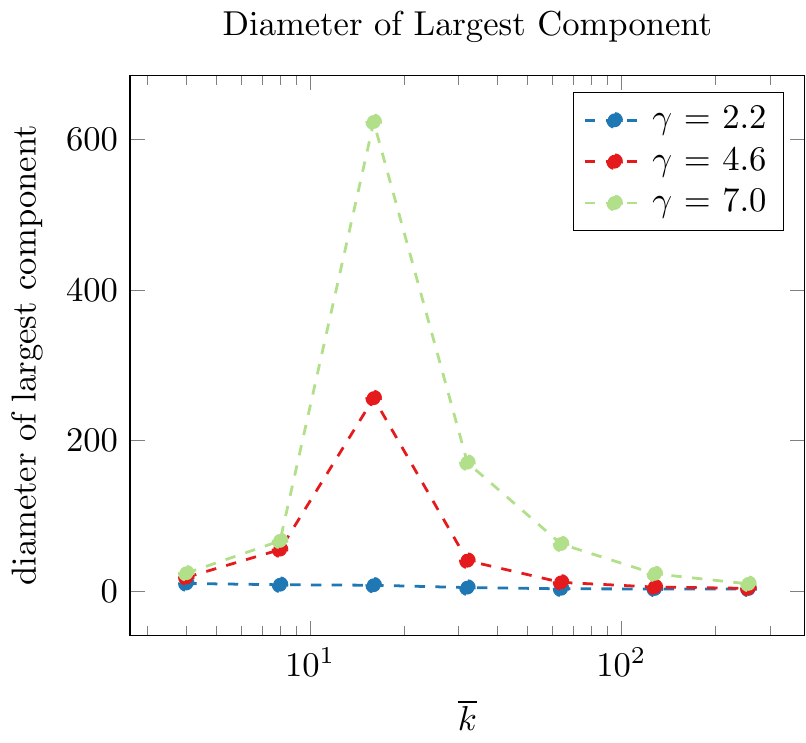}
\label{plot:native-diameter}
\end{subfigure}
\quad
\begin{subfigure}[t]{.45\linewidth}
\includegraphics[width=\linewidth]{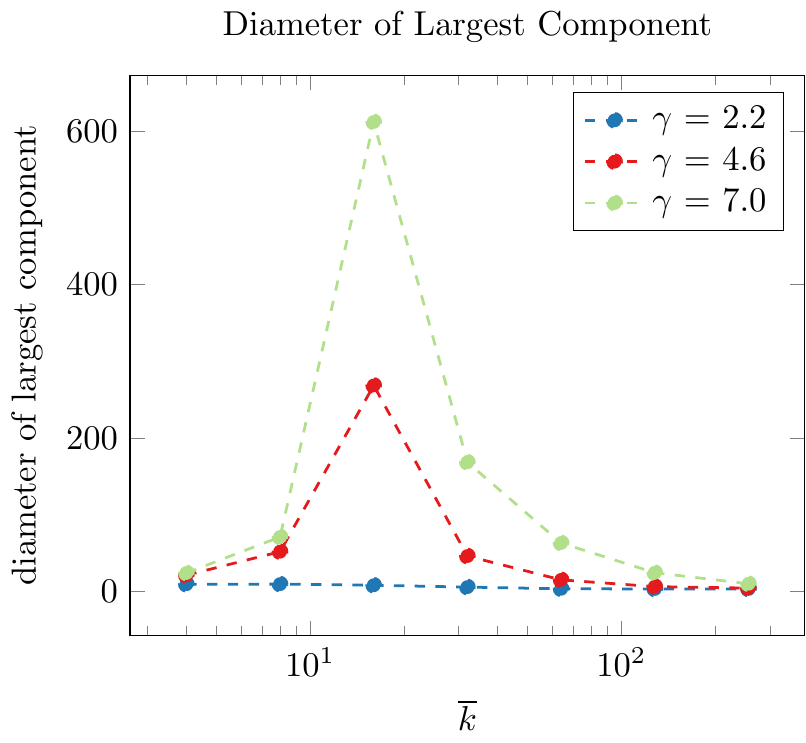}
\label{plot:comparison-diameter}
\end{subfigure}
\caption{Comparison of clustering coefficients, size of largest component and diameter of largest components for the implementation of \cite{aldecoa2015hyperbolic} (left) and our implementation (right).
Values are averaged over 10 runs.}
\label{plot:properties-comparison-I}
\end{figure}

\begin{figure}
\begin{subfigure}[t]{.45\linewidth}
\includegraphics[width=\linewidth]{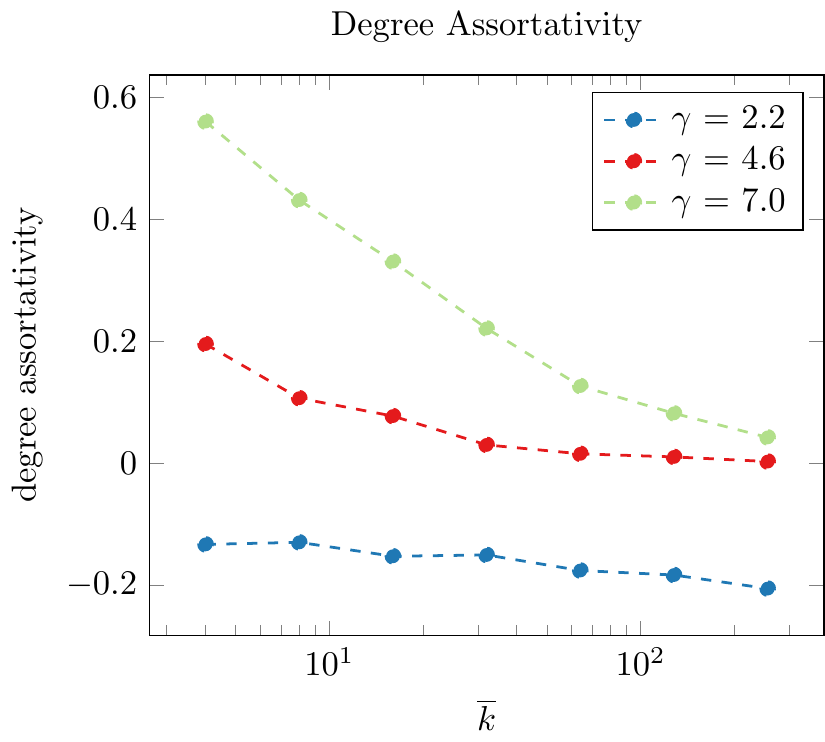}
\label{plot:native-degass}
\end{subfigure}
\quad
\begin{subfigure}[t]{.45\linewidth}
\includegraphics[width=\linewidth]{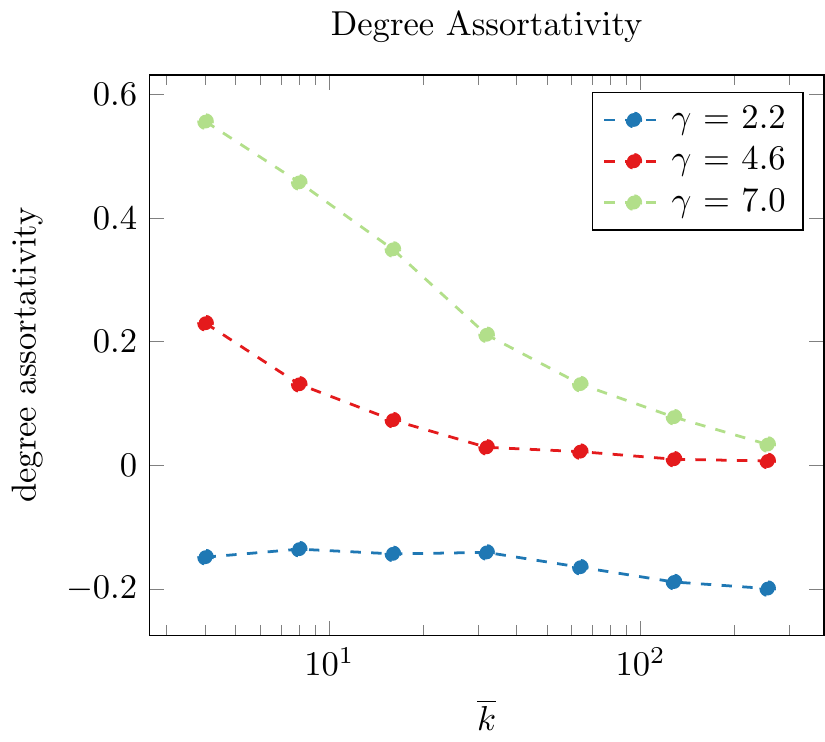}
\label{plot:comparison-degass}
\end{subfigure}

\begin{subfigure}[t]{.45\linewidth}
\includegraphics[width=\linewidth]{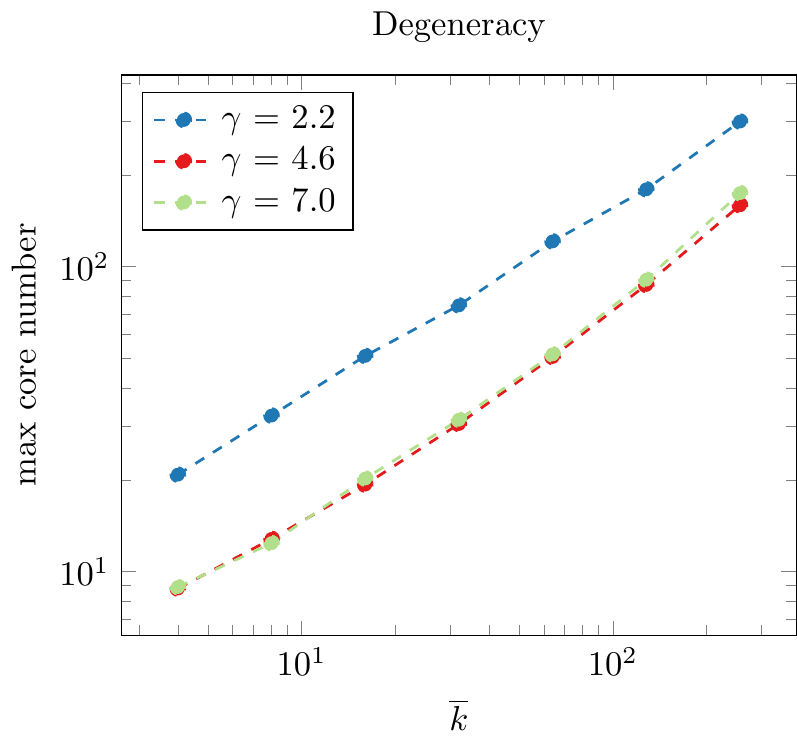}
\label{plot:native-degen}
\end{subfigure}
\quad
\begin{subfigure}[t]{.45\linewidth}
\includegraphics[width=\linewidth]{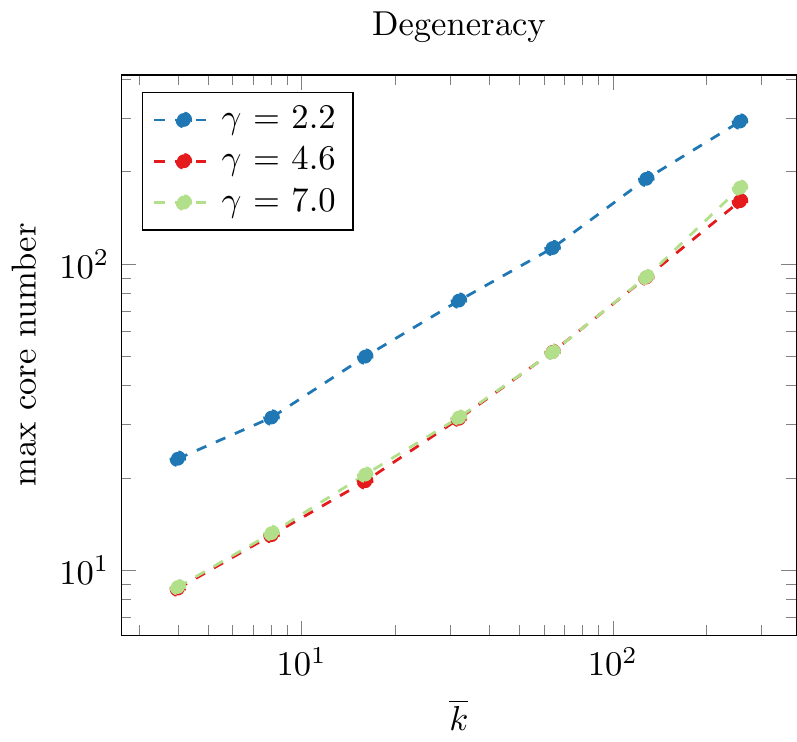}
\label{plot:comparison-degen}
\end{subfigure}

\begin{subfigure}[t]{.45\linewidth}
\includegraphics[width=\linewidth]{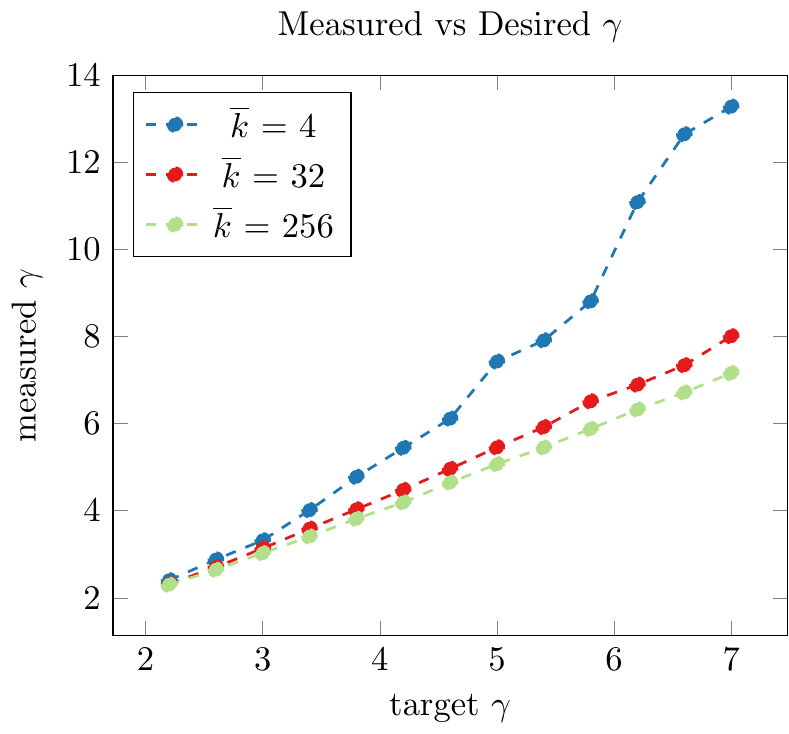}
\label{plot:native-gamma}
\end{subfigure}
\quad\begin{subfigure}[t]{.45\linewidth}
\includegraphics[width=\linewidth]{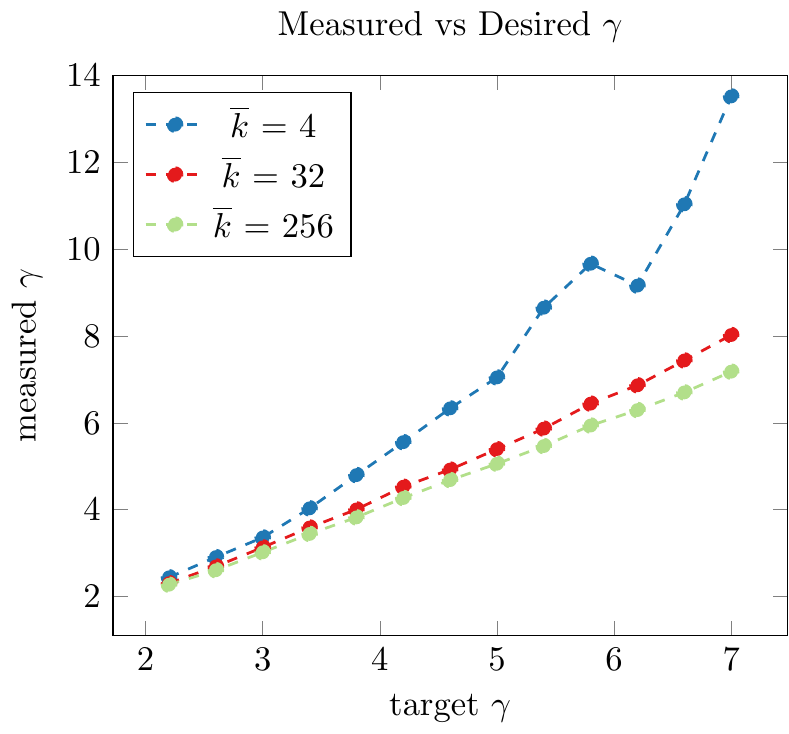}
\label{plot:comparison-gamma}
\end{subfigure}
\caption{Comparison of degree assortativity, degeneracy and measured vs desired $\gamma$ for the implementation of \cite{aldecoa2015hyperbolic} (left) and our implementation (right).
Values are averaged over 10 runs.}
\label{plot:properties-comparison-II}
\end{figure}

\begin{figure}
\begin{subfigure}[t]{.45\linewidth}
\includegraphics[width=\linewidth]{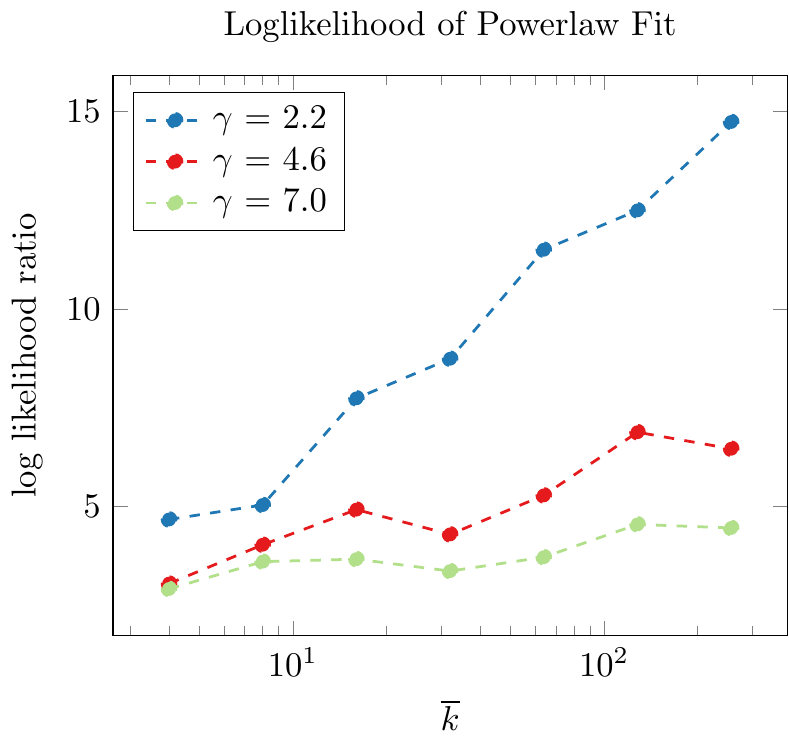}
\label{plot:native-power-law-ll}
\end{subfigure}
\quad
\begin{subfigure}[t]{.45\linewidth}
\includegraphics[width=\linewidth]{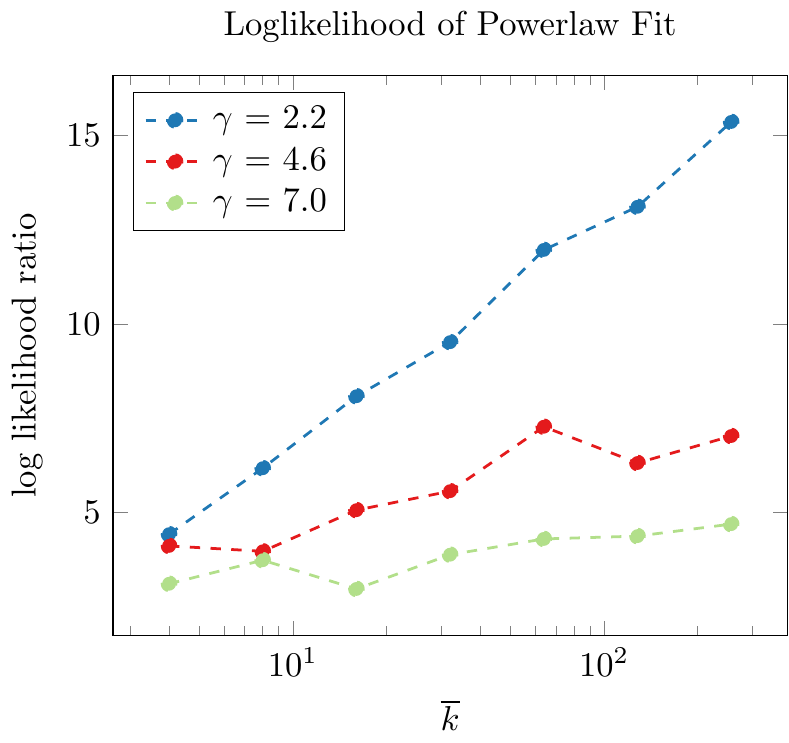}
\label{plot:comparison-power-law-ll}
\end{subfigure}
\caption{Comparison of likelihood of a power-law fit to the degree distribution for the implementation of \cite{aldecoa2015hyperbolic} (left) and our implementation (right).
Values are averaged over 10 runs.}
\label{plot:properties-comparison-III}
\end{figure}

\FloatBarrier
\section{Effect of Additional Random Edges}
\begin{figure}
\centering
\includegraphics{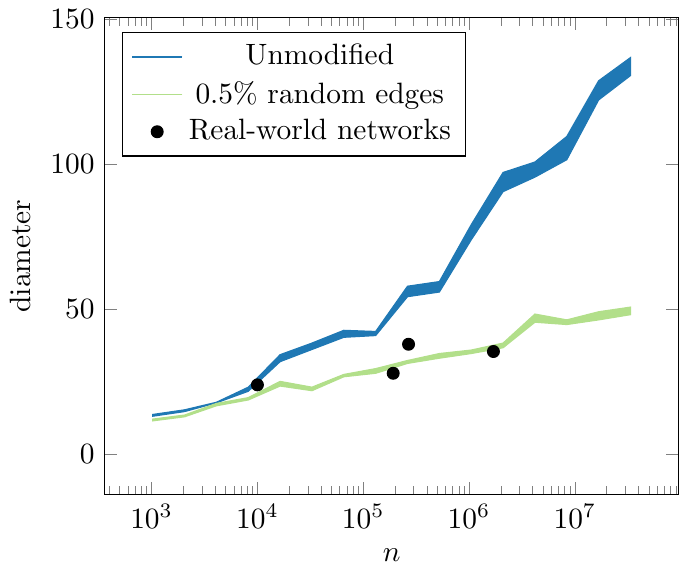}
\caption{Effect of random long-range edges on diameter ranges. Baseline graphs are generated with an average degree of 10, values are averaged over 5 runs.
Black circles correspond to PGPgiantcompo, caidaRouterLevel, citationCiteseer and as-Skitter from Table~\ref{table:real-graphs}, which have a comparable density.
}
\label{plot:long-range-diameter}
\end{figure}
\begin{figure}
\centering
\begin{subfigure}[t]{.47\linewidth}
\centering
\includegraphics{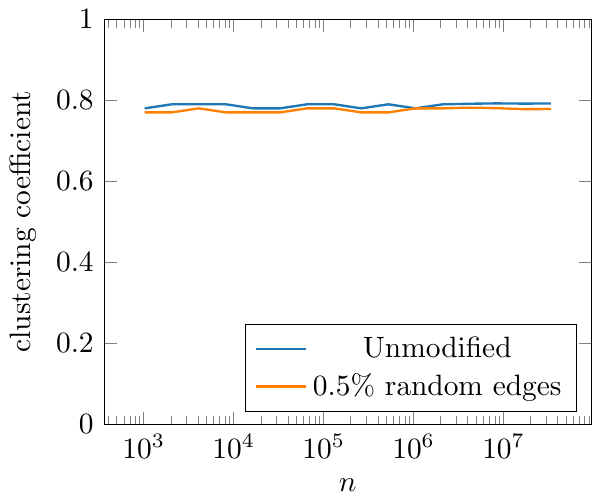}
\caption{Clustering coefficient.}
\label{plot:long-range-cc}
\end{subfigure}
\quad
\begin{subfigure}[t]{.47\linewidth}
\centering
\includegraphics{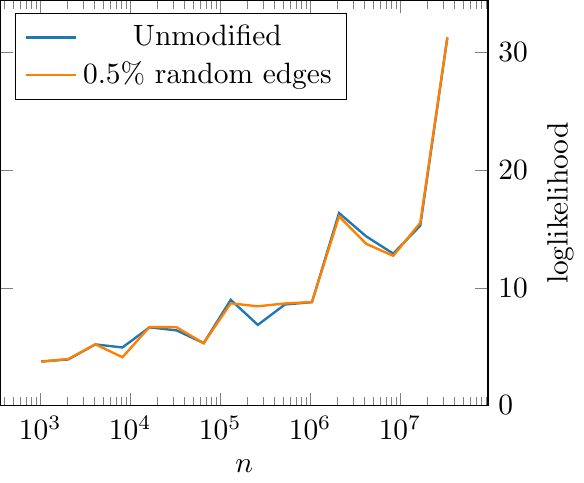}
\caption{Loglikelihood of power-law degree distribution.}
\label{plot:long-range-plll}
\end{subfigure}
\caption{Side effects of adding random edges to generated graphs. Baseline graphs are generated with an average degree of 10, values are averaged over 5 runs.
The clustering coefficient changes by less than 0.03, the two lines of the likelihood of a power-law degree distribution are nearly identical.}
\label{plot:long-range-effects}
\end{figure}
\end{document}